\documentclass[12pt]{article}

\usepackage{amssymb,amsmath,amsfonts,eurosym,geometry,ulem,graphicx,caption,color,setspace,sectsty,comment,footmisc,caption,natbib,pdflscape,subfigure,array}
\usepackage{bbold}
\usepackage{tikz}
\usepackage{pgfplots}

\pgfplotsset{width=6cm,compat=1.9}

\usepgfplotslibrary{fillbetween}
\usepackage{hyperref}

\usepackage{color}  
\hypersetup{hidelinks,colorlinks=true,linkcolor=blue,citecolor=blue}

\usepackage{natbib}
\makeatletter
\newcommand*{\textlabel}[2]{%
  \edef\@currentlabel{#1}
  \phantomsection
  #1\label{#2}
}
\makeatother

\newcolumntype{L}[1]{>{\raggedright\let\newline\\arraybackslash\hspace{0pt}}m{#1}}
\newcolumntype{C}[1]{>{\centering\let\newline\\arraybackslash\hspace{0pt}}m{#1}}
\newcolumntype{R}[1]{>{\raggedleft\let\newline\\arraybackslash\hspace{0pt}}m{#1}}

\usepackage[T1]{fontenc}
\usepackage[utf8]{inputenc}
\usepackage{setspace}
\usepackage{babel}
\usepackage{blindtext}
\usepackage[utf8]{inputenc}
\usepackage{amssymb}
\usepackage{amsthm}
\usepackage{amsmath}
\newtheorem{theorem}{Theorem}

\newtheorem{proposition}{Proposition}
\newtheorem{observation}{Observation}
\newtheorem{remark}{Remark}
\newtheorem{lemma}{Lemma}
\newtheorem{claim}{Claim}

\usepackage{array}
\usepackage{footmisc}

\newtheorem{definition}{Definition}

\title{Random Double Auction: A Robust Bilateral Trading Mechanism}
\author{Wanchang Zhang\thanks{Department of Economics, University of California, San Diego. Email: waz024@ucsd.edu. I am indebted to  Songzi Du and Joel Sobel for stimulating discussions. I thank Sarah Auster, Snehal Banerjee,  Yi-Chun Chen, Jeffrey Ely,  Simone Galperti, Xiangqian Yang and Haoxiang Zhu for helpful comments. The paper was previously circulated with the title of ``Robust Bilateral Trade Mechanisms with Known Expectations''.  }}
\date{First Draft: February 2021. Current Draft: April 2022}

\begin{document}

\maketitle
\begin{abstract}
   I construct a novel \textit{random double auction} as a robust bilateral trading mechanism for a profit-maximizing intermediary who facilitates trade between a buyer and a seller. It works as follows. The intermediary publicly commits to charging a \textit{fixed commission fee} and \textit{randomly} drawing a \textit{spread} from a uniform distribution. Then the buyer submits a bid price and the seller submits an ask price simultaneously. If the difference between the bid price and the ask price is greater than the realized spread, then the asset is transacted at the \textit{midpoint price}, and each pays the intermediary half of the fixed commission fee. Otherwise,  no trade takes place, and no one pays or receives anything. I show that the random double auction maximizes the worst-case expected profit across all \textit{dominant-strategy incentive compatible} and \textit{ex-post individually rational } mechanisms for the symmetric case. I also construct a robust trading mechanism with similar properties for the asymmetric case.

\end{abstract}
\newpage

\section{Introduction}\label{s1}
At every moment, a huge amount of trades are facilitated  by intermediaries   charging fees for their  intermediary services in matching buyers with sellers. For example, stocks are  transacted through  a trading platform that  typically gets compensation by means of  commissions;   cars are  transacted through  an automobile  dealer who  charges dealer fees;  many bonds, commodities and derivatives are transacted in the over-the-counter market (OTC) where a market maker earns profits through the bid-ask spread.\\
\indent In reality, there are many situations in which the uncertainty of the value of the asset being  traded is huge, e.g., a newly public stock, and Tesla's new model. Intermediaries may then know little of the concerned parties' willingness to trade and only have an overall estimate about it. Given the huge uncertainty towards the two-sided market, it is natural for the intermediary to seek for a trading mechanism that \textit{guarantees} a good profit.   How should a profit-maximizing intermediary design trading rules in such situations? Would the intermediary still be able to guarantee a positive profit and thus have strict incentives to offer intermediary services? \\
\indent To answer these questions, I  study the design of profit-maximizing bilateral trading mechanisms  when the intermediary has limited  information about the value distribution of the buyer and the seller. Specifically,  I assume that the intermediary knows only  the expectations of the private values of the buyer and the seller, but does not know  the joint distribution of the private values\footnote{That is, the intermediary knows neither the marginal distributions nor the correlation structure except for the expectations of the private values.}. While there are many different ways to model the intermediary's limited information and the results in this paper can be extended to models where the intermediary knows, for example,  higher moments of the  value distribution, I focus attention on the model where the intermediary only knows the expectations for two reasons. First,  the knowledge of the expectations is arguably  the minimal amount of information under which, as I will show,  one obtains a non-trivial answer. Therefore, this model can be viewed as a natural benchmark. Second, in practice,  the expectations are   often known with a higher degree of confidence compared with, for example,  higher moments. I refer to any joint distribution consistent with the known expectations as a \textit{feasible value distribution}.   The intermediary evaluates any trading mechanism by the expected profit in the worst case over all feasible value distributions, referred to as the \textit{profit guarantee}. The objective of the intermediary is to find a trading mechanism, referred to as a \textit{maxmin trading mechanism},  that maximizes the profit guarantee across all dominant-strategy
incentive compatible (DSIC) and ex-post individually
rational (EPIR) mechanisms. At a high level, this study follows the ``Wilson doctrine''  (\cite{wilson1987game}) that motivated the search for economic institutions not sensitive to unrealistic assumptions about the information structure.\\
\indent The main finding is that a \textit{random double auction} is a maxmin trading mechanism in  the symmetric\footnote{Note that the higher the seller's value, the lower his willingness to trade. Thus it is plausible to regard the highest-value seller as the lowest-type seller. When the known expectations sum up to 1, the expectation of the buyer's value and the expectation of the seller's value  have the same distance from the  lowest-type buyer and the lowest-type seller respectively. Therefore I refer to this case as the symmetric case. The  symmetric case captures situations in which  both parties have similar willingness to trade.} case in which  the known expectations of the values of the buyer and the seller sum up to the upper bound of the values, which is normalized to 1. This mechanism essentially works as follows.\\
\textbf{Step 0}: \textit{fixed commission fee}.  The intermediary publicly commits to charging a fixed commission fee $r>0$.\footnote{$r$ is characterized by the known expectations, details of which are given in Section \ref{s51}} \\
\textbf{Step 1}: \textit{uniformly random spread}. The intermediary publicly commits to randomly drawing a spread $s$ uniformly on $[r, 1]$. Then a random spread is drawn whose realization is not observed by either the buyer or the seller. The buyer and the seller both know  $r$ and the uniform distribution on $[r,1]$ from which the random spread is drawn.  \\
\textbf{Step 2}: \textit{midpoint transaction price}.  The buyer submits a bid price $b$, and the seller submits an ask price $a$, simultaneously.  If the difference between the bid price and the ask price is greater than the random spread, or $b - a > s$, then the seller sells the asset to the buyer at the midpoint price $\frac{b+a}{2}$, and each pays the intermediary half of the fixed commission fee $\frac{r}{2}$.  Otherwise, no trade takes place, and no one pays or receives anything. \\
\indent Under this mechanism, the uniformly random spread determines whether the transaction is successful or not for a bid-ask pair $(b,a)$.  That is, trade takes place randomly. Simple calculation renders that trade takes place with a probability $\frac{b-a-r}{1-r}$ when $b-a>r$ and  0 otherwise. Conditional on trading, the transaction price is the midpoint of the bid and the ask; in addition,   the intermediary earns $r$ as a total commission from both parties. Note that conditional on trading, the mechanism is a double auction, which is not strategy-proof (\cite{chatterjee1983bargaining}). The uniformly random spread makes  the mechanism strategy-proof, i.e.,  it is a dominant strategy for the buyer (resp, the seller) to submit a bid (resp, an ask) equal to his private value.\footnote{In a random double auction, the buyer's  payoff from submitting a bid $b$ when his true value is $v_B$ is $\frac{b-a-r}{1-r}\cdot (v_B-\frac{b+a+r}{2})$. It is straightforward that $b=v_B$ maximizes his  payoff regardless of the seller's submitted ask $a$. Similarly, truthful-telling maximizes the  payoff for the seller regardless of the buyer's submitted  bid $b$.} Furthermore, the uniformly random spread hedges against uncertainty about the joint distribution of values. Indeed, the uniformly random spread together with the fixed commission fee makes the intermediary indifferent to any feasible value distribution supported on value profiles where the difference between the buyer's value and the seller's value is higher than $r$, which renders the random double auction a good candidate for a maxmin trading mechanism.  To see this, note that the  profit collected from the bid-ask pair $(b,a)$ if $b-a>r$ is \[\frac{b-a-r}{1-r}\cdot r, \tag{1}\label{1}\]
where the first term is the trading probability and the second term is the fixed commission fee conditional on trading. Note that the profit is linear with respect to both the bid and the ask. Recall that the bid (resp, the ask) is equal to the true value of the buyer (resp, the seller) because the mechanism is strategy-proof. Together with the knowledge of the expectations of the values, the random double auction thus yields the same expected profit for any aforementioned joint distribution.\\ \indent Notably, the features of the  random double auction are familiar in the real world. First, there is extensive evidence that the bid-ask spread of a given financial asset is not constant, but varies over time\footnote{There are alternative explanations on the variations of the bid-ask spread.  For example, the bid-ask spread is affected by the inventory risk (\cite{ho1983dynamics}),  the trading volume (\cite{lee1993spreads}), and the asymmetric information (\cite{hasbrouck1988trades}).}, e.g., the New York Stock Exchange (NYSE) stocks and Chicago Board Options Exchange (CBOE) options. Second,    brokerage firms often adopt the  fixed-commission practice, e.g.,   Interactive Brokers offers fixed-commission plans  for many financial assets\footnote{Interactive Brokers is a brokerage firm. From its official website (interactivebrokers.com), it offers a fixed-commission plan that  charges \$0.005 per share for stocks in US; it also offers a fixed-commission plan that charges \$ 0.065 per contract for NANOS Options on CBOE.}; E*TRADE charges a fixed commission per contract for futures contracts\footnote{E*TRADE is also a brokerage firm. From its official website (us.etrade.com), it   charges \$1.5 per contract  for futures contracts. }. Third,   a double auction is widely used in stock exchanges as well as in dark pools\footnote{A dark pool is a privately organized financial forum or exchange for trading securities that are not accessible by the investing public. Dark pools came about primarily to facilitate block trading involving a huge number of securities.}, e.g.,  the New York Stock
Exchange (NYSE) and the Tokyo Stock Exchange (TSE) use a double auction to determine the opening prices;  block-trading dark pools such as Liquidnet or POSIT typically match orders at
the midpoint of the prevailing bid-ask prices (\cite{duffie2017size}). To my knowledge, the random double auction is a \textit{novel} trading mechanism that  combines the above three commonly observed features.  \\
\indent To show that the random double auction is a maxmin trading mechanism, I  reformulate  the intermediary's  problem into a zero-sum game between the intermediary and 
adversarial nature who chooses a feasible  value distribution to minimize the expected profit, and then I construct a  feasible value distribution, referred to as a \textit{worst-case value distribution}, such that the random double auction and the worst-case value distribution form a Nash equilibrium  of the zero-sum game. Precisely,  (i) the worst-case value distribution minimizes the expected profit over all feasible value distributions  under the random double auction, and (ii) the random double auction maximizes  the expected profit over all DSIC and EPIR trading mechanisms under the worst-case value distribution. Indeed, (i) implies that the random double auction's profit guarantee is the expected profit when adversarial natures chooses the worst-case value distribution, and (ii) implies that no other DSIC and EPIR trading mechanism can yield a strictly higher expected profit under the worst-case value distribution. Hence, (i) and (ii) together imply the random double auction is a maxmin trading mechanism.  \\
\indent The worst-case value distribution is a \textit{symmetric triangular value distribution} that  can be described as follows. The support is a symmetric triangular subset in the set of joint valuations,  which is the same as the trading region\footnote{I refer to the set of value profiles in which trade takes place with a positive probability as the trading region.} of the random double auction. The marginal distribution for the buyer is a combination of a  uniform distribution on $(r,1)$ and an atom on 1, while for the seller is a combination of a  uniform distribution on $(0,1-r)$ and an atom on 0.  The conditional distribution is some truncated generalized Pareto distribution with an atom on 1 (resp, 0) for the buyer (resp, the seller). As I will show, the symmetric triangular value distribution and the   random double auction form a saddle point (Theorem \ref{t1}).\\
\indent Let me briefly illustrate the construction of the symmetric triangular value distribution. As the intermediary is \textit{indifferent} between trading and no trading for any value profile in the trading region of the random double auction except for the value profile where the value of the buyer (resp, the seller) is 1 (resp, 0),  I impose a condition on the joint distribution requiring that  the  ``virtual value''\footnote{In this setting, the ``virtual value'' of a value profile is  that the difference between the buyer's value and the seller's value  minus the sum of their information rents that are pinned down by  dominant-strategy incentive compatibility and the binding ex-post participation
constraints of zero-value buyer and one-value seller. This is  a straightforward adaption of Myerson's virtual value. The details are given in Section \ref{s44}.  } be \textit{zero} for any one of these value profiles, from which I obtain the symmetric triangular value distribution. \\
\indent Indeed, this condition guarantees that the intermediary is indifferent to  any DSIC and EPIR trading mechanism in which 1) ex-post participation constraints are binding for zero-value buyer and one-value seller, and 2) trade does not take place if the value profile lies outside  the support and trade takes place with probability one when the value of the buyer (resp, the seller) is 1 (resp, 0). In addition,  such a trading mechanism is an optimal trading mechanism  given the symmetric triangular value distribution. Interestingly, the symmetric triangular value distribution exhibits positive correlation: if the buyer's value is higher, then the seller's value is more likely to be higher as well. Intuitively, positive correlation levels the maximal gain from trade across value profiles and therefore limits the intermediary's incentive to discriminate across value profiles. Indeed, this particular value distribution renders the intermediary indifferent across all value profiles in the support but one in which  the  value of the buyer (resp, the seller) is 1 (resp, 0).\\
\indent To  show that the symmetric triangular value distribution is a worst-case value distribution under the random double auction, I use the linear programming duality theorem. I consider  adversarial nature's problem of seeking for a worst-case value distribution and derive its dual program. It turns out that the complementary slackness condition holds under the random double auction by \eqref{1}.\\
\indent I extend my analysis to constructing a maxmin trading mechanism for the asymmetric case in which the known expectations of the values of the buyer and the seller sum up to a number other than 1 (Theorem \ref{t2}). I find that the maxmin trading mechanism  shares similar properties to the random double auction in that 1) trade is restricted:  trade does not take place if the difference between a \textit{weighted} bid price and a \textit{weighted} ask price  falls below some threshold; 2)  trade takes place randomly otherwise, albeit with a different randomization device\footnote{The exact trading rule of the  maxmin trading mechanism for the asymmetric case is given in Section \ref{s52}.}.  I follow the saddle point approach for the result. More specifically, I construct a saddle point of the zero-sum game by solving  a similar complementary slackness condition and  a similar 0-virtual-value condition.  \\ 
\indent Randomized trading is a salient property of the constructed mechanisms, for both the symmetric case and the asymmetric case. This requires the intermediary to have  \textit{full commitment power}, which is a standard assumption in the mechanism design literature (e.g., \cite{myerson1981optimal}). However, in practice, it is hard for  traders to check whether the randomization is done according to the specified trading rule. Traders then may not trust the specified randomization. This motivates the search for a  trading mechanism  that maximizes the profit guarantee across all \textit{deterministic} DSIC and EPIR trading mechanisms. Such a trading mechanism is referred to as a   \textit{maxmin deterministic trading mechanism}.   I characterize   the  class of maxmin deterministic trading mechanisms for any pair of expectations with a non-trivial profit guarantee (Theorem \ref{t3}). Examples of  maxmin deterministic trading mechanisms include a \textit{linear trading mechanism}, in which trade takes place with probability one if and only if the difference between a weighted bid price and a weighted ask price exceeds a threshold, and a \textit{double posted-price trading mechanism}, in which trade takes place with probability one if and only if the bid price exceeds a threshold and the ask price falls short of a threshold.\\
\indent In addition, I extend my result to a more general model in which the intermediary can hold the asset. That is, the sum of the two traders' allocations is only required to be weakly less than 1. I show that the random double auction (resp, the constructed trading mechanism in Theorem \ref{t2}) remains a maxmin trading mechanism for the symmetric case (resp, the asymmetric case) (Theorem \ref{t4}). Finally, I  apply my result to an information design problem in which a financial regulator can choose a probability distribution of the  value profile of the buyer and the seller to maximize their welfare. The intermediary,
after observing the choice of the distribution but not the realized joint valuations, designs a
profit-maximizing trading mechanism across DSIC and EPIR trading mechanism. I show that a special symmetric triangular value distribution is a solution to this financial regulator's information design problem (Theorem \ref{t6}).\\ 
\indent The remainder of  the introduction discusses the related literature.  Section \ref{s3} presents the model.  Section \ref{s4} illustrates the methodology.  Section \ref{s5} characterizes the main  results. Section \ref{s6} characterizes the class of maxmin deterministic trading mechanisms. Section \ref{s7} extends and discusses the main results.  Section \ref{s8} is a conclusion. All proofs are in the Appendix. 
\subsection{Related Literature}\label{s2}
This paper contributes to the literature of robust mechanism design.  \cite{carrasco2018optimal}  study the one-dimensional profit-maximizing selling mechanisms when the seller   faced with a single buyer only knows the first $N$ moments of the distribution of the buyer's value ($N$ can be any positive integer), and solve the problem in which the seller only 
knows the expectation of the buyer's value as a special case. In contrast,  this paper studies  multi-dimensional trading  mechanisms when the intermediary knows the expectations of the traders' values.   Similarly,  both papers employ duality theory  and  construct worst-case value distributions to proceed the analysis. Notably, the worst-case value distribution in this paper has a rather intricate correlation structure exhibiting a particular positive correlation, while the worst-case value distribution in that paper is single-dimensional. \\
\indent \cite{koccyiugit2020distributionally} consider a  model of auction design in which adversarial nature chooses the worst-case value distribution subject to symmetric expectations of bidders' values. They find, among other results, that a highest-bidder lottery mechanism is optimal within the DSIC and EPIR mechanisms in which only the highest bidder can be allocated the good. In contrast, I consider a model of bilateral trade and do not place any  restriction on DSIC and EPIR mechanisms or any restriction on the known expectations for the main results. \cite{suzdaltsev2020optimal}  considers a model of auction design in which the auctioneer  only knows the expectation of 
each bidder’s value, but focuses on  deterministic  mechanisms. He shows that a linear version of Myerson's optimal auction is a maxmin mechanism across all deterministic DSIC and EPIR mechanisms. Although the class of deterministic mechanisms is not the main focus of this paper, I characterize maxmin deterministic DSIC and EPIR mechanisms  in a model of bilateral trade (Theorem \ref{t3}).\\
\indent \cite{zhang2021correlation} considers a model of auction design in which the auctioneer  knows the marginal distribution of each bidder's value but does not know the correlation structure. He finds, among other results, that the second-price auction with the uniformly distributed random reserve is a maxmin auction across DSIC and EPIR mechanisms under certain regularity conditions for the two-bidder case. In contrast, this paper considers a model of bilateral trade and assumes that the intermediary only knows the expectations of the traders' values. Methodologically, both papers construct  worst-case value distributions to proceed the analysis. However, as will be clear later, the construction of the worst-case value distribution is more involved in this paper: it requires me to solve a partial integral equation in addition to ordinary differential equations. In addition, in contrast to that paper,  this paper offers a complete characterization for all primitives. \\
\indent \cite{brooks2021maxmin} consider   informationally robust  auction design in the interdependent value environment. They assume that the auctioneer only knows the expectation of each bidder's value, but does not know the correlation structure or  (common-prior) beliefs. They construct a maxmin mechanism across Bayesian mechanisms that  maximizes minimum profit across all correlation and
information structures and all equilibria.  In contrast, my framework assumes that values are known to the agents, and restricts attention to dominant-strategy mechanisms, ruling out issues brought by higher order beliefs. More specifically, dominant-strategy mechanisms are robust to misspecification of agents'  beliefs, and are thus more appropriate for situations in which one can not say much about agents' beliefs. \\
\indent There are other papers seeking robustness to value distributions, e.g.,  \cite{carrasco2019robust}, \cite{auster2018robust}, \cite{bergemann2011robust}, \cite{bergemann2008pricing}.
\cite{carroll2017robustness}, \cite{giannakopoulos2020robust} and \cite{chen2019distribution} focus on  the problem
of selling multiple goods to a single buyer when the value distributions are unknown. A separate strand of papers  focuses on the case in which the designer does not have reliable
information about the agents’ hierarchies of beliefs about each other while assuming
the knowledge of the payoff environment, e.g., \cite{bergemann2005robust}, \cite{chung2007foundations}, \cite{chen2018revisiting}, 
\cite{bergemann2016informationally,bergemann2017first,bergemann2019revenue}, \cite{du2018robust}, 
\cite{brooks2021optimal}, \cite{libgober2021informational}, \cite{yamashita2018foundations}.   \cite{carroll2019robustness} provides an elaborate survey on  various notions of robustness studied in the  literature, e.g., robustness to preferences, robustness to strategic behavior and robustness to interaction
among agents.\\
\indent There are other papers studying robust bilateral trading mechanisms. \cite{wolitzky2016mechanism}  studies optimal bilateral trading mechanisms in terms of efficiency when agents are maxmin expected utility maximizers, with similar ambiguity sets (that is, the buyer knows only the expectation of the seller’s value, and
vice versa). \cite{bodoh2012ambiguous} also assumes that the agents are maxmin expected utility maximizers, but focuses on profit-maximizing bilateral trading mechanisms in the applications. The main difference from these two papers is that my paper addresses robustness concerns on the part of mechanism designer instead of on the part of economic agents.  \cite{carroll2016informationally} studies bilateral trading mechanisms within the informationally robust framework with a focus on the expected surplus. In contrast, my paper considers a private-value environment and focuses on profit-maximizing mechanism design.\\
\indent There are papers in the computer science literature that study the performance of simple bilateral trading mechanisms. \cite{deng2021approximately} show that  the better of two simple posted price mechanisms guarantees at least 10\% of the first-best gains-from-trade in the independent private value environment. \cite{dutting2021efficient} show, among others, that an adjusted VCG mechanism using a single sample from the seller's distribution  guarantees 50\% of the first-best welfare in the independent private value environment. 
\section{Preliminaries}\label{s3}
I consider an environment where an asset is traded between two risk-neutral traders through an intermediary. One of the traders is the seller ($S$), who holds the asset initially, while the other one is the buyer ($B$), who does not hold the asset initially.  I denote by $I = \{S,B\}$ the set of traders and $i\in I$ is a trader. Each trader $i$ has private information about his
value for the asset, which is modeled as a random variable $v_i$ with cumulative distribution
function $F_i$.\footnote{I do not make any assumption on the distributions of these random variables. That is,  the distributions could be continuous, discrete, or any mixtures. In addition,  $F_S$ could be different from $F_B$. } I use $f_i(v_i)$ to denote
the density of $v_i$ in the distribution $F_i$ when $F_i$ is differentiable at $v_i$; I use $Pr_i(v_i)$ to denote
the probability of $v_i$ in the distribution $F_i$ when $F_i$ has a probability mass at $v_i$.  I denote  by $V_i$ the set of possible values of trader $i$.  Throughout, I assume $V_S=V_B$. I assume that $V_i$ is bounded.  As a normalization, I assume that $V_i = [0,1]$.  The set of possible value profiles 
is denoted by $V = [0,1]^2$ with a typical value profile $v$. The joint distribution of the value profile is denoted by $F$. The set of all joint distributions on  $V$ is denoted by $\Delta V$.\\
\indent  The intermediary only knows the expectations $M_B$ and $M_S$  of the private values of $B$ and $S$ respectively as well as the set of possible value profiles $V$,  but does not know the joint distribution $F$. Formally,  I denote by
$$\Pi(M_B,M_S) = \{
\pi \in \Delta V : \int v_B \pi(v)dv = M_B,\int v_S \pi(v)dv = M_S\}$$
the collection of joint distributions that are consistent with the known expectations. I refer to any $\pi \in \Pi(M_B, M_S)$ as a \textit{feasible value distribution}.\\
\indent The
intermediary  seeks a  dominant-strategy
incentive compatible (DSIC) and ex-post individually
rational (EPIR) trading mechanism.   A direct trading mechanism\footnote{As I focus on dominant-strategy trading mechanisms, the revelation principle holds and it is without loss of generality to restrict attention to direct trading mechanisms.} $(q,t_B,t_S)$ consists of a trading rule $q : V \to [0, 1]$, a payment rule $t_B : V\to \mathbb{R}$ and a transfer rule $t_S: V\to \mathbb{R}$.\footnote{$q$ is the probability that the buyer obtains the asset when the asset is indivisible. I allow randomization, which will play a crucial role in my analysis.  $q$ can be interpreted as the trading quantity when  the asset is divisible. } The buyer submits a bid price $b$ and the seller submits an ask price $a$ simultaneously to the intermediary. Upon receiving the bid-ask pair $(b,a)$, the buyer obtains the asset with probability $q(b,a)$ and pays $t_B(b,a)$ to the intermediary, while the seller holds the good with the remaining probability $1-q(b,a)$ and receives  $t_S(b,a)$ from the intermediary. With slight abuse of notation, I sometimes use the true value profile $v=(v_B,v_S)$ to represent the submitted bid-ask pair because each trader truthfully reports his value in the dominant-strategy equilibrium. The set of all DSIC and EPIR trading mechanisms is denoted by $\mathcal{D}$.  Formally,
$(q,t_B,t_S)\in \mathcal{D}$ if 
\[v_Bq(v)-t_B(v)\ge v_Bq(v_B',v_S)-t_B(v_B',v_S),  \quad \forall  v\in V,v_B'\in V_B ;\tag{$DSIC_B$}\]
\[v_Bq(v)-t_B(v)\ge 0, \quad \forall v\in V;  \tag{$EPIR_B$}\]
\[v_S(1-q(v))+t_S(v)\ge v_S(1-q(v_B,v_S'))+t_S(v_B,v_S'),\quad \forall v\in V,v_S'\in V_S; \tag{$DSIC_S$}\]
\[v_S(1-q(v))+t_S(v)\ge v_S,\quad \forall v\in V. \tag{$EPIR_S$}\]

\indent I am interested in the intermediary's \textit{expected profit} in the dominant-strategy equilibrium in which each trader truthfully reports his value of the asset.  The expected profit of a DSIC and EPIR trading mechanism $(q,t_B,t_S)$ under the joint distribution $\pi$ is $U((q,t_B,t_S),\pi)= \int_{v\in V}t(v) d\pi(v)$ where $t(v)=t_B(v)-t_S(v)$, referred to as the \textit{ex-post profit}. The intermediary evaluates a trading mechanism by its worst-case expected profit over all feasible value distributions. Formally, the intermediary evaluates a trading mechanism $(q,t_B,t_S)$ by its \textit{profit guarantee} $PG((q,t_B,t_S))$, defined as \[\inf_{\pi\in \Pi(M_B,M_S)}U((q,t_B,t_S),\pi). \tag{PG}\label{PG}\]   The intermediary aims to find a trading mechanism $(q^*,t_B^*,t_S^*)$, referred to as a \textit{maxmin trading mechanism}, that maximizes the profit guarantee. Formally, the intermediary solves \[
\sup_{(q,t_B,t_S)\in \mathcal{D}}PG((q,t_B,t_S)).\tag{MTM}\label{mtm}\]

\section{Preparations for Main Results}\label{s4}
\subsection{Zero-Sum Game}
 I observe that the intermediary's maxmin optimization problem (\ref{mtm}) can be interpreted as a two-player sequential \textit{zero-sum
game}. The two players are the intermediary and  adversarial nature. The intermediary first chooses
a trading mechanism $(q,t_B,t_S)\in \mathcal{D}$. After observing the intermediary’s choice of
the trading mechanism, adversarial nature chooses a feasible value distribution $\pi\in \Pi(M_B,M_S)$. The intermediary’s
payoff is $U((q,t_B,t_S),\pi)$, and adversarial nature’s payoff is $-U((q,t_B,t_S),\pi)$. Instead
of solving directly for such a subgame perfect equilibrium,  I can solve for a Nash equilibrium 
$((q^*,t_B^*,t_S^*),\pi^*)$ of the simultaneous-move version of this zero-sum game, which corresponds to a saddle point of the payoff functional $U$. I refer to $\pi^*$ as a \textit{worst-case value distribution}.  Formally, for  any $(q,t_B,t_S)\in \mathcal{D}$ and any $\pi\in \Pi(M_B,M_S)$,\[U((q^*,t_B^*,t_S^*), \pi) \ge U((q^*,t_B^*,t_S^*),\pi^*) \ge U((q,t_B,t_S),\pi^*).\tag{SP}\] As discussed in the introduction, the properties of a saddle point imply that the intermediary’s
equilibrium strategy in the simultaneous-move game, $(q^*,t_B^*,t_S^*)$, is also her maxmin strategy (i.e. her
equilibrium strategy in the subgame perfect equilibrium of the sequential game). I will construct a saddle point for either of the main results.
\subsection{Revenue Equivalence}
To facilitate the analysis, it will be useful  to further simplify the problem.  I will use the following proposition: its proof is standard
but included in Appendix \ref{aa} for completeness.
\begin{proposition}[Revenue Equivalence]\label{p1}
When searching for a maxmin trading mechanism, it is without loss to restrict attention to trading mechanisms satisfying the  following properties:\\(i). $q(v)$ is non-decreasing in $v_B$ and non-increasing in $v_S$.\\(ii). $t_B(v)=v_Bq(v)-\int_0^{v_B}q(x,v_S)dx$.\\(iii).
    $t_S(v)=v_Sq(v)+\int_{v_S}^1q(v_B,x)dx$.\\(iv).
    $t(v)=(v_B-v_S)q(v)-\int_0^{v_B}q(x,v_S)dx-\int_{v_S}^1q(v_B,x)dx$.
\end{proposition}
That is,  the trading rule $q(v)$ is monotone;   the payment rule $t_B$ and the transfer rule $t_S$ admit an envelope representation. In addition,  the ex-post participation constraints for zero-value buyer and one-value seller are binding.  Proposition \ref{p1} is standard in the mechanism design literature, and it simplifies the problem because the ex-post profit can be expressed using the trading rule and thus I can focus attention on the trading rule for my analysis. 
\subsection{Complementary Slackness Condition}
Consider adversarial nature's problem (\ref{PG}). That is, fixing an arbitrary trading mechanism $(q,t_B,t_S)\in\mathcal{D}$,  adversarial nature seeks for a feasible value distribution $\pi$ that minimizes the expected profit under $(q,t_B,t_S)$. I observe that the problem (\ref{PG}) is a semi-infinite dimensional linear program. I derive its dual program.  By Theorem 3.12 in \cite{anderson1987linear}, I establish strong duality (details can be found in Appendix \ref{aa}), which implies the following lemma.
\begin{lemma}\label{l1}
 Given a trading  mechanism $(q,t_B,t_S)\in \mathcal{D}$, if $\pi$ minimizes the expected profit over $\Pi(M_B,M_S)$,  then there exist real numbers $\lambda_B$, $\lambda_S$ and $\mu$ such that 
\[
\lambda_Bv_B+\lambda_Sv_S+\mu = t(v),\quad \forall v \in supp(\pi).\tag{CS}\label{cs}\]
\end{lemma}
(\ref{cs}) is  a \textit{complementary slackness condition}, stating that the ex-post profit is a linear function of the true values  for any value profile in the support of a worst-case value distribution. We will see that this complementary slackness condition plays a crucial role in the construction of a maxmin trading mechanism. 
\subsection{Weighted Virtual Value}\label{s44}
Now consider the problem that  fixing any joint distribution $\pi$\footnote{For exposition, I assume that $\pi$ is differentiable everywhere when deriving the weighted virtual values. It can be easily extended to joint distributions which admits an atom on the value profile $(1,0)$.}, the intermediary seeks for  a trading mechanism $(q,t_B,t_S)\in \mathcal{D}$ that maximizes the expected profit under $\pi$. I denote, with slight abuse of notation,  the joint density function of $\pi$ by $\pi(v_i,v_j)$,  the marginal density function by $\pi_i(v_i)= \int_{v_j}\pi(v_i,v_j)dv_j$ , the density function of $v_i$ conditional on $v_j$ by $\pi_i(v_i| v_j)=\frac{\pi(v_i,v_j)}{\pi_j(v_j)}$ and the cumulative  distribution function of $v_i$ conditional on $v_j$ by $\Pi_i(v_i|v_j)= \int_0^{v_i} \pi_i(x| v_j)dx$. I define $\Pi_i(v_i,v_j):= \pi_j(v_j)\Pi_i(v_i|v_j)=\int_0^{v_i} \pi(x,v_j)dx$. A direct implication of Proposition \ref{p1} is that the expected profit of $(q,t)$ under the joint distribution $\pi$ is \[E[t(v)]=\int_vq(v)\Phi(v)dv,\]
where \[\Phi(v)= 
\pi(v)(v_B-v_S)-(\pi_S(v_S)-\Pi_B(v_B,v_S))- \Pi_S(v_S,v_B). \]
Here $\Phi(v)$ is referred to as the \textit{weighted  virtual value}, which is the ``virtual value'' \footnote{Note that $\Phi(v)=\pi(v)(v_B-\frac{1-\Pi_B(v_B|v_S)}{\pi_B(v_B|v_S)} -(v_S+\frac{\Pi_S(v_S|v_B)}{\pi_S(v_S|v_B)}))$ when $\pi(v)$ is not 0. Here $\phi(v):= v_B-\frac{1-\Pi_B(v_B|v_S)}{\pi_B(v_B|v_S)} -(v_S+\frac{\Pi_S(v_S|v_B)}{\pi_S(v_S|v_B)})$ is the ``virtual value'' of the value profile $(v_B,v_S)$ in this environment, which is the difference between  the conditional virtual value of the buyer and that of  the seller. However, it turns out that the weighted virtual value is more convenient for my analysis because it is well-defined even if $\pi(v)=0$. Henceforth,  I directly work with the weighted virtual value. } weighted by the probability that the value profile is $v$.   Thus,  the problem of designing an optimal trading mechanism given a joint distribution can be viewed as maximizing the inner product of the trading rule and the weighted virtual values  subject to that the trading rule is monotone. We will see that the weighted virtual value plays a crucial role in the construction of a worst-case value distribution. 
\section{Main Results}\label{s5}
I focus on the case in which $0<M_S<M_B<1$ throughout the paper. If $M_S=0$, it is straightforward to see that the  problem \eqref{mtm} is equivalent to a monopolistic seller's problem of selling a good to a single buyer when the sellers only knows the expectation of the buyer's value, which has been studied by  \cite{carrasco2018optimal}. Likewise, if $M_B=1$, the problem \eqref{mtm} is reduced to a single-dimensional problem and the solution can be similarly characterized, details of which is relegated to Appendix \ref{af}.  If $M_B\le M_S$, then the problem \eqref{mtm} becomes trivial as adversarial nature can always choose a distribution such that the seller's value is never below the buyer's value, e.g., a joint distribution that puts all probability masses on  the value profile $(M_B,M_S)$. Thus, the profit guarantee can not be positive for any $(q,t_B,t_S)\in \mathcal{D}$, implied by EPIR. Then, trivially, the Never Trading  Mechanism (i.e., $q(v)=t_S(v)=t_B(v)=0$ for any $v\in V$) is a maxmin trading mechanism. I will call a pair of  expectations \textit{non-trivial} if $M_B>M_S$.
\subsection{Symmetric Case: $M_B+M_S=1$}\label{s51}
Recall the random double auction: given a submitted bid-ask pair $(b,a)$, if $b-a>s$ where $s$ is a random spread drawn from the uniform distribution on $[r,1]$, then trade takes place at the midpoint  price $\frac{b+a}{2}$, and each pays the intermediary $\frac{r}{2}$; otherwise,  trade does not take place, and no one pays or receives anything. The fixed commission fee $r$ is determined by the known expectations: \[r=1-\sqrt{1-(M_B-M_S)}.\] 
It is straightforward to show that the random double auction can also be expressed as follows. If $b-a>r$,  \[
q^*(b,a)=
    \frac{1}{1-r}(b-a-r),\]
    \[t^*_B(b,a)=\frac{1}{2(1-r)}[b^2-(a+r)^2],\]
    \[t^*_S(b,a)=\frac{1}{2(1-r)}[(b-r)^2-a^2].\]
If $b-a\le r$, 
\[q^*(b,a)=t^*_B(b,a)=t^*_S(b,a)=0.\] 
Interestingly,  the trading rule is  a \textit{linear} function;   the payment rule and the transfer rule are both  \textit{quadratic} functions. In addition, this mechanism satisfies the standard weak budget balance property (as in \cite{myerson1983efficient}), i.e., the intermediary never subsidize the market. \\
\indent Now let me specify the symmetric triangular value distribution. The support is a symmetric triangular subset of joint valuations $ST:=\{v\in V|v_B-v_S> r\}$.  The marginal distribution for the buyer is a combination of a uniform distribution on $(r,1)$ and an atom of size $r$ on 1:   $\pi^*_B(v_B)=1$ for $v_B\in (r,1)$ and $Pr^*_B(1)=r$. The marginal distribution for the seller is a combination of a uniform distribution on $(0,1-r)$ and an atom of size $r$ on 0:  $\pi^*_S(v_S)=1$ for $ v_S\in (0,1-r)$ and $Pr^*_S(0)=r$. The conditional distribution for the buyer is a combination of some generalized Pareto distribution on $(v_S+r,1)$ and an atom on 1:  when $v_S\in (0,1-r)$, $\pi^*_B(v_B|v_S)=\frac{2r^2}{(v_B-v_S)^3}$ for $v_B\in (v_S+r, 1)$ and $Pr_B^*(v_B=1|v_S)=\frac{r^2}{(1-v_S)^2}$; when $v_S=0$, $\pi^*_B(v_B|v_S=0)=\frac{r}{(v_B)^2}$ for $v_B\in (r, 1)$ and $Pr_B^*(v_B=1|v_S=0)=r$. The conditional distribution for the seller is a combination of some generalized Pareto distribution on $(0,v_B-r)$ and an atom on 0:  when $v_B\in (r,1)$, $\pi^*_S(v_S|v_B)=\frac{2r^2}{(v_B-v_S)^3}$ for $v_S\in (0, v_B-r)$ and $Pr_S^*(v_S=0|v_B)=\frac{r^2}{(v_B)^2}$; when $v_B=1$, $\pi^*_S(v_S|v_B=1)=\frac{r}{(1-v_S)^2}$ for $v_S\in (0, 1-r)$ and $Pr_S^*(v_S=0|v_B=1)=r$. 
Equivalently, the symmetric triangular value distribution can be described as a combination of a joint density function on $ST \backslash  \{(1,0)\}$ and an atom of size $r^2$ on the value profile $(1,0)$ as follows (See Figure \ref{fig:S1}).
\[\pi^*(v_B,v_S)=  \left\{
\begin{array}{lll}
\frac{2r^2}{(v_B-v_S)^3}     &      & {\normalfont\text{if $v_B-v_S > r$, $v_B\neq 1$ and $v_S\neq 0$},}\\
\frac{r^2}{(1-v_S)^2}    &      & {\normalfont\text{if $ v_B=1$ and $0<v_S < 1-r$},}\\
\frac{r^2}{(v_B)^2}    &      & {\normalfont\text{if $r< v_B < 1$ and $v_S=0$}.}
\end{array} \right. \]
\[Pr^*(1,0)=r^2.\]

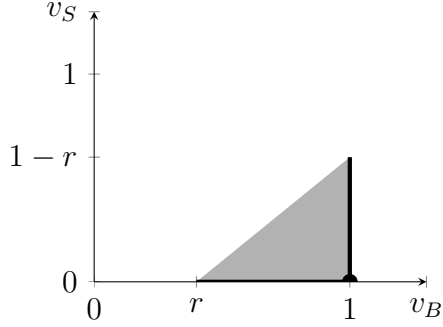
\begin{figure}
\centering
\begin{tikzpicture}
\begin{axis}[
    axis lines = left,
    xmin=0,
        xmax=1.3,
        ymin=0,
        ymax=1.3,
        xtick={0,0.4,1,1.3},
        ytick={0,0.6,1,1.3},
        xticklabels = {$0$, $r$, $1$, $v_B$},
        yticklabels = {$0$, $1-r$, $1$, $v_S$},
        legend style={at={(1.1,1)}}
]

\path[name path=axis] (axis cs:0,0) -- (axis cs:1,0);
\path[name path=A] (axis cs: 0.4,0) -- (axis cs:1,0.6);
\path[name path=B] (axis cs:0,0.4) -- (axis cs:1,0.4);
\path[name path=C] (axis cs:0,0) -- (axis cs:1,0);
\addplot[area legend, black!30] fill between[of=A and C,  soft clip={domain=0.4:1}];
\addplot[black, ultra thick] coordinates {(0.4, 0) (1, 0)};
\addplot[black, ultra thick] coordinates {(1, 0.6) (1, 0)};
\node[black,right] at (axis cs:0.97,0.04){};
\node at (axis cs:1,0) [circle, scale=0.5, draw=black,fill=black] {};
\end{axis}
\end{tikzpicture}
\caption{Symmetric Triangular Value Distribution} \label{fig:S1}
\end{figure}
\begin{definition}[Positive Correlation for Bivariate Distributions]
\normalfont  Let $Z = (X,Y)$ be
a bivariate random vector whose distribution is $F$. I say that $Z$ exhibits positive correlation for  $D_X$ and $D_Y$ if $F(X|Y=y)$ \textit{first order stochastically dominates} $F(X|Y=y')$ for any $y>y', y, y'\in D_Y$ and $F(Y|X=x)$ \textit{first order stochastically dominates} $F(Y|X=x')$ for any $x>x', x, x'\in D_X$.
\end{definition}
\begin{remark}
\normalfont The symmetric triangular value distribution exhibits positive correlation for $r< v_B<1$ and $0<v_S< 1-r$.\footnote{To see this, note that $\Pi^*_S(v_S|v_B)=\frac{r^2}{(v_B-v_S)^2}$ is decreasing w.r.t. $v_B$ for $v_B\in (r,1)$. When $v_B=1$,  $\Pi^*_S(v_S|v_B=1)=\frac{r}{1-v_S}\ge \frac{r^2}{(1-v_S)^2}$, so the positive correlation breaks when $v_B=1$. Similarly, $\Pi^*_B(v_B|v_S)=1-\frac{r^2}{(v_B-v_S)^2}$ is decreasing w.r.t. $v_S$ for $v_S\in (0,1-r)$. When $v_S=0$, $\Pi^*_B(v_B|v_S=0)=1-\frac{r}{v_B}\le 1-\frac{r^2}{(v_B)^2}$, so the positive correlation breaks when $v_S=0$.}
\end{remark}
\begin{theorem}\label{t1}
For the symmetric case, the random double auction is a maxmin trading mechanism with a profit guarantee of $[1-\sqrt{1-(M_B-M_S)}]^2$, and the symmetric triangular value distribution is a worst-case value distribution.
\end{theorem}
\begin{remark}
\normalfont It is useful to compare the revenue guarantee of the random double auction with the optimal revenue across DSIC and EPIR trading mechanisms if the value distribution were known to the intermediary. One interesting case could be the following value distribution: the buyer's value follows a uniform distribution on $[2M_B-1,1]$ and the seller's value follows a uniform distribution on $[0,2M_S]$; their values are independent. By a straightforward adaptation of the revenue equivalence theorem, the revenue achievable by the optimal dominant-strategy mechanism can be computed. For example, When $M_S=\frac{1}{8}$ and $M_B=\frac{7}{8}$, the optimal revenue is $\frac{1}{2}$, whereas the revenue guarantee of the random double auction is $\frac{1}{4}$, so the ratio between the revenue guarantee and the optimal revenue is $\frac{1}{4}$. In addition, this ratio is large when $M_B-M_S$ is large and converges to 1 as $M_B-M_S\to 1$. Another interesting case could be that the value distribution is a point mass on $(M_B,M_S)$.  Then the optimal revenue is $M_B-M_S$. When $M_S=\frac{1}{8}$ and $M_B=\frac{7}{8}$, the ratio between the revenue guarantee and the optimal revenue is $\frac{1}{3}$.  In addition, this ratio is increasing in $M_B-M_S$ and converges to 1 as $M_B-M_S\to 1$. 
\end{remark}
\subsubsection{Characterization of Random Double Auction}
I  form an educated guess of the trading region in a maxmin trading mechanism: trade takes place with positive probability if and only if the difference between the bid (true value of the buyer) $v_B$ and the ask (true value of the seller) $v_S$ exceeds a threshold $r>0$, or $v_B-v_S>r$.  This threshold $r$ is analogous to the reserve price in auction (e.g., \cite{myerson1981optimal}). In the bilateral trading environment, the intermediary exercises her monopoly power by restricting trade to raise the expected profit.   In addition, the support of a worst-case value distribution coincides with the trading region (including the boundary). Together with $(iv)$ of Proposition \ref{p1}, the complementary slackness condition (\ref{cs}) can be expressed as follows: for any $(v_B,v_S)\in ST$,  \[
\lambda_B^*v_B+\lambda_S^*v_S+\mu^* = (v_B-v_S)q^*(v)-\int_{v_S+r}^{v_B}q^*(x,v_S)dx-\int_{v_S}^{v_B-r}q^*(v_B,x)dx. \tag{CS-1}\label{eq3}\] 
Now I solve for  the trading rule $q^*$.  First I  take the first order derivatives with respect to $v_B$ and $v_S$ respectively, and I obtain that for any $(v_B,v_S)\in ST$,
\[
(v_B-v_S)\frac{\partial q^*(v_B,v_S)}{\partial v_B}-\frac{\partial \int_{v_S}^{v_B-r}q^*(v_B,x)dx }{\partial v_B}=\lambda_B^*,\tag{FOC-B-1}\label{eq4}\]
\[ (v_B-v_S)\frac{\partial q^*(v_B,v_S)}{\partial v_S}-\frac{\partial \int_{v_S+r}^{v_B}q^*(x,v_S)dx }{\partial v_S}=\lambda_S^*. \tag{FOC-S-1}\label{eq5}\]
Then, I take the cross partial derivative, and  I obtain that for any $(v_B,v_S)\in ST$,
\[
(v_B-v_S)\frac{\partial q^{*}(v_B,v_S)}{\partial v_B\partial v_S}=0.\]
Thus, $q^{*}(v_B,v_S)$ is separable, which can be expressed  as the sum of two functions $f^{*}$ and $g^{*}$: for any $(v_B,v_S)\in ST$,
\[q^*(v_B,v_S)=f^*(v_B)+g^*(v_S). \tag{4.1.1}\label{eq7}\]
The separable nature of $q^{*}(v_B,v_S)$ is crucial for solving \eqref{eq3}. Plugging \eqref{eq7} into \eqref{eq4} and \eqref{eq5}, I obtain that for any $(v_B,v_S)\in ST$,
\[\label{eq8}
r(f^*)'(v_B)-(f^*(v_B)+g^*(v_B-r))=\lambda_B^*,\tag{4.1.2}
\]
\[\label{eq9}
r(g^*)'(v_S)+f^*(v_S+r)+g^*(v_S)=\lambda_S^*.\tag{4.1.3}
\]
Note that $f^*(v_B)+g^*(v_B-r)=0$ and that $f^*(v_S+r)+g^*(v_S)=0$ because trade does not take place in the boundary of the trading region, i.e.,  $q^*(v_B,v_S)=0$ if $v_B-v_S=r$.  Then it is straightforward to  solve for  $f^*(v_B)$ and $g^*(v_S)$, and I obtain that  for any $(v_B,v_S)\in ST$,
\[\label{eq10}
    f^*(v_B)=\frac{\lambda_B^*}{r}v_B+c_B^*,\tag{4.1.4}
\]
\[\label{eq11}
    g^*(v_S)=\frac{\lambda_S^*}{r}v_S+c_S^*,\tag{4.1.5}
\]
where $c_B^*$ and $c_S^*$ are some constants. Again, using $q^*(v_B,v_S)=0$ if $v_B-v_S=r$,  I obtain that 
\[\label{eq12}
   \lambda_S^*=c_B^*+c_S^*=-\lambda_B^*. \tag{4.1.6}
\]
Now plugging \eqref{eq10},\eqref{eq11} and \eqref{eq12} into \eqref{eq7}, I obtain that for any $(v_B,v_S)\in ST$,
\[\label{eq13}
    q^*(v_B,v_S)=\frac{\lambda_B^*}{r}(v_B-v_S-r).
\]
Finally, to solve for $\lambda_B^*$, I let  $q^*(1,0)$ be $1$ and obtain  that  $\lambda_B^*=\frac{r}{1-r}$. This is because the value profile $(1,0)$ is the ``highest''  joint type \footnote{In this environment, a high-value buyer and a low-value seller are more willing to trade. Thus, a value profile with a high buyer's value and a low seller's value can be interpreted as a ``high joint type'' in the traditional mechanism design literature.} in this environment in that it has the highest virtual value of 1\footnote{Indeed, trade has to take place with probability one for this value profile in any maxmin trading mechanism because the worst-case value distribution puts a probability mass on this value profile.}.  So far I have obtained the trading rule $q^*$\footnote{Plugging the trading rule $q^*$ into \eqref{eq3}, it is straightforward that $\mu^*=-\frac{r^2}{1-r}$.}. 
The payment rule $t_B^*$ (resp, the transfer rule $t_S^*$) is then characterized by $(ii)$ (resp, $(iii)$) of Proposition \ref{p1}. 
\subsubsection{Characterization of Symmetric Triangular Value Distribution}
I impose a \textit{0-weighted-virtual-value condition} on the joint distribution, stating that weighted virtual value is 0 for any value profile in the support except for the highest joint type. Formally, \[\label{eq16}
  \Phi(v)=0,  \quad \forall v\in ST\backslash\{(1,0)\}.\tag{ZWVV-1}
\]
\indent I now illustrate how I construct a feasible value distribution such that \eqref{eq16} holds.  I start from  value profiles in which either $v_B=1$ or $v_S=0$. Assume that $Pr^*(1,0)=\alpha$. Consider  value profiles $(v_B,0)$ in which $v_B\in (r,1)$. Let  $S^*(v_B,0):= \int_{(v_B,1)}\pi^*(x,0)dx+Pr^*(1,0)$ for $v_B\in (r,1)$ and $S^*(1,0):= Pr^*(1,0)$.  Note that  $\pi^*(v_B,0)=-\frac{\partial S^*(v_B,0)}{\partial v_B}$ for $v_B \in (r,1)$. By \eqref{eq16}, I have that for any $(v_B,0)$ in which $v_B\in (r,1)$, 
\[\label{eq18}
    \pi^*(v_B,0)(v_B-0)-S^*(v_B,0)=0.
\]
Note that this is a simple ordinary differential equation, to which the solution is 
\[\label{eq19}
    S^*(v_B,0)=\frac{\alpha}{v_B},\quad \pi^*(v_B,0)=\frac{\alpha}{v^2_B}, \quad \forall v_B\in (r,1).
\]
Then consider value profiles $(1,v_S)$ in which $ v_S \in (0,1-r)$. Similarly, let $S^*(1,v_S):= \int_{(0,v_S)}\pi^*(1,x)dx+Pr^*(1,0)$ for $v_S\in (0,1-r)$. Note that  $\pi^*(1,v_S)=\frac{\partial S^*(1,v_S)}{\partial v_S}$ for $v_S\in (0,1-r)$. By \eqref{eq16}, I have that for any  $(1,v_S)$ in which $ v_S \in (0,1-r)$, 
\[\label{eq20}
    \pi^*(1,v_S)(1-v_S)-S^*(1,v_S)=0.
\]
Note that this is also a  simple ordinary differential equation, to which the solution is 
\[\label{eq21}
    S^*(1,v_S)=\frac{\alpha}{1-v_S},\quad \pi^*(1,v_S)=\frac{\alpha}{(1-v_S)^2},\quad \forall v_S\in (0,1-r).
\]
Finally consider  any value profile $(v_B,v_S)$ in which $v_B-v_S > r$, $v_B\neq 1$ and $v_S\neq 0$. Let $S^*(v_B,v_S):= \int_{(v_B,1)}\pi^*(b,v_S)db+\pi^*(1,v_S)$ if  $v_B-v_S > r$, $v_B\neq 1$ and $v_S\neq 0$.  Note that $\pi^*(v_B,v_S)=-\frac{\partial S^*(v_B,v_S)}{\partial v_B}$ if $v_B-v_S > r$, $v_B\neq 1$ and $v_S\neq 0$. By \eqref{eq16}, I have that if $v_B-v_S > r$, $v_B\neq 1$ and $v_S\neq 0$,
\[\label{eq22}
    \pi^*(v_B,v_S)(v_B-v_S)-S^*(v_B,v_S)-\int_{(0,v_S)}\pi^*(v_B,s)ds-\pi^*(v_B,0)=0 \tag{PIE}
\]
Note  that \eqref{eq22} is a (second order) partial integral equation. It is straightforward to see that  $S^*(v_B,v_S)$ is not separable by taking the cross partial derivative. I take the guess-and-verify approach to solve  \eqref{eq22}. I guess that if $v_B-v_S > r$, $v_B\neq 1$ and $v_S\neq 0$,
\[\label{eq23}
    S^*(v_B,v_S)=\frac{\alpha}{(v_B-v_S)^2}.
\]
Under this guess,  the L.H.S. of \eqref{eq22} equals $\frac{2\alpha}{(v_B-v_S)^3}(v_B-v_S)-\frac{\alpha}{(v_B-v_S)^2}-
\int_{(0,v_S)}\frac{2\alpha}{(v_B-s)^3}ds-\frac{\alpha}{v^2_B}$, which can be shown to be 0 with simple algebra. Thus, I verified the guess.\\
\indent To solve for $\alpha$, I use the requirement that $\pi^*(v)$ is a distribution. Note that the marginal distribution for $S$ is $\pi^*_S(v_S)=S^*(v_S+r,v_S)=\frac{\alpha}{(v_S+r-v_S)^2}=\frac{\alpha}{r^2}$ for $0<v_S< 1-r$ and $Pr^*_S(v_S=0)=S^*(r,0)=\frac{\alpha}{r}$. Since the integration is 1, I obtain that 
\[\label{eq24}
    \frac{\alpha}{r}+\frac{\alpha}{r^2}\cdot (1-r)=1.
\]
Thus, $\alpha=r^2$.\\ 
\indent  The final step is  to make sure that the constructed joint distribution has the known expectations, which will allow me to solve for  $r$. Formally, 
\[\label{eq25}
    r\cdot 1+\int_r^1v_B dv_B=M_B,\tag{KE-B-1}
\]
\[\label{eq26}
    r\cdot 0+\int_0^{1-r}v_S dv_S=M_S.\tag{KE-S-1}
\]
Summing up \eqref{eq25} and \eqref{eq26}, I obtain that $M_B+M_S=1$, which is the symmetric case I am considering. It is straightforward to show that $r=1-\sqrt{1-(M_B-M_S)}$ is the unique solution.

\subsection{Asymmetric Case: $M_B+M_S\neq 1$}\label{s52}
I follow the same approach to study the asymmetric case. Indeed, the characterization for the symmetric case provides us with good intuitions about the solution for the asymmetric case. I will  propose a \textit{logarithmic trading mechanism}\footnote{The name of this trading mechanism is motivated by that the trading rule, the payment rule and the transfer rule all involve  logarithmic functions.} and an \textit{asymmetric triangular value distribution}, and then show that they form a saddle point. The illustration of the result is relegated to Appendix \ref{ae}. \\
\indent The logarithmic trading mechanism $(q^{**},t_B^{**},t_S^{**})$ is described as follows. If $ r_2b-(1-r_1)a> r_1r_2$, 
\[q^{**}(b,a)=\frac{1}{\ln\frac{1-r_2}{r_1}}[\ln(\frac{1-r_1-r_2}{1-r_1}b+\frac{r_1r_2}{1-r_1})-\ln(\frac{1-r_1-r_2}{r_2}a+r_1)],\]
\[t^{**}_B(b,a)=-\frac{r_1r_2}{(1-r_1-r_2)\ln\frac{1-r_2}{r_1}}[\ln(\frac{1-r_1-r_2}{1-r_1}b+\frac{r_1r_2}{1-r_1})-\ln(\frac{1-r_1-r_2}{r_2}a+r_1)]\]\[
+\frac{1}{\ln\frac{1-r_2}{r_1}}(b-\frac{1-r_1}{r_2}a-r_1),\]
\[t^{**}_S(b,a)=-\frac{r_1r_2}{(1-r_1-r_2)\ln\frac{1-r_2}{r_1}}[\ln(\frac{1-r_1-r_2}{1-r_1}b+\frac{r_1r_2}{1-r_1})-\ln(\frac{1-r_1-r_2}{r_2}a+r_1)]
\]\[+\frac{1}{\ln\frac{1-r_2}{r_1}}(\frac{r_2}{1-r_1}b-a-\frac{r_1r_2}{1-r_1}).\]
Here  $(r_1,r_2)$ in which $r_1\in (0,1)$, $r_2\in (0,1)$ and $r_1+r_2\neq 1$,  as will be shown in Lemma \ref{l2},  is a solution to the following system of equations:
\[\label{eq27}
    M_B=\int_{r_1}^1\frac{r_1(1-r_2)}{(\frac{1-r_1-r_2}{1-r_1}v_B+\frac{r_1r_2}{1-r_1})^2}v_Bdv_B+r_1:= H_1(r_1,r_2),\tag{KE-B-2}
\]
\[\label{eq28}
    M_S=\int_{0}^{r_2}\frac{r_1(1-r_2)}{(\frac{1-r_1-r_2}{r_2}v_S+r_1)^2}v_Sdv_S:= H_2(r_1,r_2).\tag{KE-S-2}
\]
If $ r_2b-(1-r_1)a\le  r_1r_2$,  \[q^{**}(b,a)=t^{**}_B(b,a)=t^{**}_S(b,a)=0.\]
\begin{remark}
\normalfont If $M_S=0$, then it is common knowledge that the seller's value is 0. Note that $q^{**}(b,0)=\frac{1}{\ln\frac{1-r_2}{r_1}}\ln(\frac{1-r_1-r_2}{r_1(1-r_1)}b+\frac{r_2}{1-r_1})$. If $r_2=0$, it is straightforward that it is reduced to the mechanism found by \cite{carrasco2018optimal} when the monopolistic seller only knows the expectation of the buyer's value.
\end{remark}
\begin{remark}
\normalfont The logarithmic trading mechanism also satisfies the standard weak budget balance property.
\end{remark}
\indent The logarithmic trading mechanism can be implemented as follows, similar to the random double auction. \\
\textbf{Step 0}: \textit{transformed bid and ask}.  The intermediary publicly commits to  transforming a bid price $b$ and an ask price $a$ as follows: $b'=\frac{1}{\ln\frac{1-r_2}{r_1}}[\ln(\frac{1-r_1-r_2}{1-r_1}b+\frac{r_1r_2}{1-r_1})], a'=\frac{1}{\ln\frac{1-r_2}{r_1}}[\ln(\frac{1-r_1-r_2}{r_2}a+r_1)]$. The buyer and the seller both know $r_1$ and $r_2$ as well as the transformations. \\
\textbf{Step 1}: \textit{uniformly random spread}. The intermediary publicly commits to randomly drawing a spread $s'$ uniformly on $[0, 1]$. Then a random spread is drawn whose realization is not observed by either the buyer or the seller. The buyer and the seller both know   the uniform distribution on $[0,1]$ from which the random spread is drawn.  \\
\textbf{Step 2}: \textit{random transaction price and random commission fee}.  The buyer submits a bid price $b$, and the seller submits an ask price $a$, simultaneously.  If the difference between the transformed bid price  and the transformed ask price is greater than the realized spread, or $b' - a' > s'$, then the seller sells the asset to the buyer at the  price $\frac{r_2}{1-r_1-r_2}[(\frac{1-r_2}{r_1})^{a'+s'}-r_1]$, and the buyer pays the intermediary the commission fee $(\frac{1-r_2}{r_1})^{a'+s'}$.  Otherwise, no trade takes place, and no one pays or receives anything. \\
\indent Now let me specify the asymmetric triangular value distribution. The support is an asymmetric triangular subset of joint valuations $AT:=\{v|r_2v_B-(1-r_1)v_S> r_1r_2\}$.   The marginal distribution for the buyer is a combination of some generalized Pareto distribution on $(r_1,1)$ and an atom of size $r_1$ on 1: $\pi^{**}_B(v_B)=\frac{r_1(1-r_2)}{(\frac{1-r_1-r_2}{1-r_1}v_B+\frac{r_1r_2}{1-r_1})^2}$ for $ v_B \in (r_1,1)$ and $Pr^{**}_B(1)=r_1$.  The marginal distribution for the buyer is a combination of some generalized Pareto distribution on $(0,r_2)$ and an atom of size $1-r_2$ on 0: $\pi^{**}_S(v_S)=\frac{r_1(1-r_2)}{(\frac{1-r_1-r_2}{r_2}v_S+r_1)^2}$ for $v_S\in (0,r_2)$ and  $Pr^{**}_S(0)=1-r_2$. The conditional distribution for the buyer is a combination of some (different) generalized Pareto distribution on $(r_1+\frac{1-r_1}{r_2}v_S, 1)$ and an atom on $1$: when $v_S\in (0,r_2)$, $\pi^{**}_B(v_B|v_S)=\frac{2(\frac{1-r_1-r_2}{r_2}v_S+r_1)^2}{(v_B-v_S)^3}$ for $v_B\in (r_1+\frac{1-r_1}{r_2}v_S, 1) $ and $Pr^{**}_B(v_B=1|v_S)=\frac{(\frac{1-r_1-r_2}{r_2}v_S+r_1)^2}{(1-v_S)^2}$; when $v_S=0$, $\pi^{**}_B(v_B|v_S=0)=\frac{r_1}{(v_B)^2}$ for $v_B\in (r_1, 1) $ and $Pr^{**}_B(v_B=1|v_S=0)=r_1$. The conditional distribution for the seller is a combination of some (different) generalized Pareto distribution on $(0, \frac{r_2(v_B-r_1)}{1-r_1})$ and an atom on 0: when $v_B\in (r_1,1)$,  $\pi^{**}_S(v_S|v_B)=\frac{2(\frac{1-r_1-r_2}{1-r_1}v_B+\frac{r_1r_2}{1-r_1})^2}{(v_B-v_S)^3}$ for $v_S\in (0, \frac{r_2(v_B-r_1)}{1-r_1})$ and $Pr^{**}_B(v_S=0|v_B)=\frac{(\frac{1-r_1-r_2}{1-r_1}v_B+\frac{r_1r_2}{1-r_1})^2}{(v_B)^2}$; when $v_B=1$, $\pi^{**}_S(v_S|v_B=1)=\frac{1-r_2}{(1-v_S)^2}$ for $v_S\in (0, r_2) $ and $Pr^{**}_S(v_S=0|v_B=1)=1-r_2$.  Equivalently, the asymmetric triangular value distribution can be described as a combination of a joint density function on $AT \backslash  \{(1,0)\}$ and an atom of size $r_1(1-r_2)$ on the value profile $(1,0)$ as follows (See Figure \ref{fig:S2}).
$$\pi^{**}(v_B,v_S)=  \left\{
\begin{array}{lll}
\frac{2r_1(1-r_2)}{(v_B-v_S)^3}     &      & {\text{if $r_2v_B-(1-r_1)v_S\ge r_1r_2$, $v_B\neq 1$ and $v_S\neq 0$},}\\
\frac{r_1(1-r_2)}{(1-v_S)^2}    &      & {\text{if $v_B=1$ and $0<v_S < r_2$},}\\
\frac{r_1(1-r_2)}{(v_B)^2}   &      & {\text{if $r_1< v_B < 1 $ and $v_S=0$}.}
\end{array} \right. $$
$$Pr^{**}(1,0)=r_1(1-r_2).$$
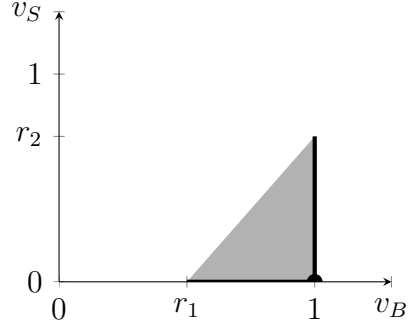
\begin{figure}
\centering
\begin{tikzpicture}
\begin{axis}[
    axis lines = left,
    xmin=0,
        xmax=1.3,
        ymin=0,
        ymax=1.3,
        xtick={0,0.5,1,1.3},
        ytick={0,0.7,1,1.3},
        xticklabels = {$0$, $r_1$, $1$, $v_B$},
        yticklabels = {$0$, $r_2$, $1$, $v_S$},
        legend style={at={(1.1,1)}}
]

\path[name path=axis] (axis cs:0,0) -- (axis cs:1,0);
\path[name path=A] (axis cs: 0.5,0) -- (axis cs:1,0.7);
\path[name path=B] (axis cs:0,0.5) -- (axis cs:1,0.5);
\path[name path=C] (axis cs:0,0) -- (axis cs:1,0);
\addplot[area legend, black!30] fill between[of=A and C,  soft clip={domain=0.5:1}];
\addplot[black, ultra thick] coordinates {(0.5, 0) (1, 0)};
\addplot[black, ultra thick] coordinates {(1, 0.7) (1, 0)};
\node[black,right] at (axis cs:0.97,0.04){};
\node at (axis cs:1,0) [circle, scale=0.5, draw=black,fill=black] {};
\end{axis}
\end{tikzpicture}
\caption{Asymmetric Triangular Value Distribution} \label{fig:S2}
\end{figure}

\begin{remark}
\normalfont The asymmetric triangular value distribution exhibits positive correlation for $r_1< v_B<1$ and $0<v_S< r_2$.\footnote{To see this, note that $\Pi^{**}_S(v_S|v_B)=\frac{(\frac{1-r_1-r_2}{1-r_1}v_B+\frac{r_1r_2}{1-r_1})^2}{(v_B-v_S)^2}$ is decreasing w.r.t. $v_B$ for $v_B \in (r_1,1)$. When $v_B=1$, $\Pi^{**}_S(v_S|v_B)=\frac{1-r_2}{1-v_S}\ge \frac{(1-r_2)^2}{(1-v_S)^2}$, so  the positive correlation breaks when $v_B=1$. Similarly, $\Pi^{**}_B(v_B|v_S)=1-\frac{(\frac{1-r_1-r_2}{r_2}v_S+r_1)^2}{(v_B-v_S)^2}$ is decreasing w.r.t. $v_S$ for $v_S\in (0,r_2)$. When $v_S=0$, $\Pi^{**}_B(v_B|v_S=0)=1-\frac{r_1}{v_B}\le 1-\frac{(r_1)^2}{(v_B)^2}$, so the positive correlation breaks when $v_S=0$.}
\end{remark}
\begin{remark}
\normalfont  If $M_S=0$, then it is common knowledge that the seller's value is 0. When $v_S=r_2=0$, it is straightforward that the asymmetric triangular value distribution is reduced to the worst-case distribution found by \cite{carrasco2018optimal} when the monopolistic seller only knows the expectation of the buyer's value.
\end{remark}
\begin{lemma}\label{l2}
For the asymmetric case, there exists a solution $(r_1,r_2)\in (0,1)^2$ to the system of equations \eqref{eq27} and \eqref{eq28}. In addition, $r_1+r_2\neq 1$. 
\end{lemma}
\begin{theorem}\label{t2}
For the asymmetric case, the logarithmic trading mechanism is a maxmin trading mechanism with a profit guarantee of  $r_1(1-r_2)$, and the asymmetric triangular value distribution is a worst-case value distribution.
\end{theorem}
\begin{remark}
\normalfont If $M_B+M_S\to 1$, it is straightforward to show that there is a solution in which $r_1+r_2\to 1$. Then by L'H\^opital's rule, the logarithmic trading mechanism converges to the random double auction (that is,  $q^{**}\to q^*$, $t_B^{**}\to t_B^*$, and $t_S^{**}\to t_S^*$). In addition, it is straightforward that the asymmetric triangular value distribution converges to the the symmetric triangular value distribution.
\end{remark}

\section{Deterministic Mechanisms}\label{s6}
In this section, I restrict attention to the class of  deterministic DSIC and EPIR trading mechanisms and  characterize maxmin deterministic trading mechanisms across mechanisms in this class.  Note that  Proposition \ref{p1} still holds, with an additional property that $q(v)$ \footnote{For exposition, I define $q(v)$ to be 0 if $v\notin V$. } is either 0 or 1 for any $v\in V$. Recall that I focus on the case in which $0<M_S<M_B<1$.
\begin{definition}
\normalfont The \textit{trade boundary} of a given deterministic DSIC and EPIR trading mechanism $(q,t_B,t_S)$ is a set of value profiles  $\mathcal{B}:=\{\bar{v}=(\bar{v}_B,\bar{v}_S)\in V| q(\bar{v})=0\footnote{For exposition, I assume that trade does not take place on the trade boundary. As will be clear,  this is to guarantee that a best response for adversarial nature exists. This assumption does not affect the solution and the value of the problem.  Similar assumption is also made  in \cite{kos2015selling}.} \quad \text{and\indent for any small}\quad \epsilon>0, \quad q(\bar{v}_B+\epsilon,\bar{v}_S)=1\quad \text{or}\quad q(\bar{v}_B,\bar{v}_S-\epsilon)=1\}$.
\end{definition}
\indent I observe that  the trade boundary of a  deterministic DSIC and EPIR trading mechanism is non-decreasing (See Figure \ref{fig:S5}).
\begin{observation}
\normalfont If $\bar{v}=(\bar{v}_B, \bar{v}_S)\in \mathcal{B},\bar{v}'= (\bar{v}'_B, \bar{v}'_S) \in \mathcal{B}$ and $\bar{v}_B>\bar{v}_B'$, then $\bar{v}_S\ge \bar{v}_S'$. \footnote{To see this, note that  by the definition of the trade boundary, I have that $q(\bar{v}_B,\bar{v}_S')=1$ because $\bar{v}' \in \mathcal{B}$ and $\bar{v}_B>\bar{v}_B'$. Then,  again by the definition of the trade boundary, I have that $\bar{v}_S\ge \bar{v}_S'$ because $\bar{v}\in \mathcal{B}$.}
\end{observation}
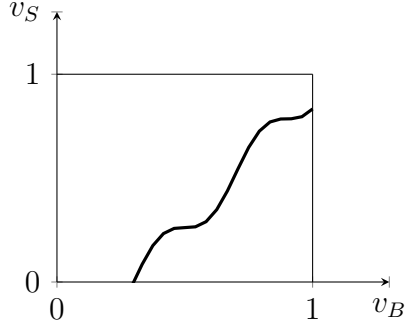
\begin{figure}
\centering
\begin{tikzpicture}
\begin{axis}[
    axis lines = left,
    xmin=0,
        xmax=1.3,
        ymin=0,
        ymax=1.3,
        xtick={0,1,1.3},
        ytick={0, 1,1.3},
        xticklabels = {$0$, $1$, $v_B$},
        yticklabels = {$0$,  $1$, $v_S$},
        legend style={at={(1.1,1)}}
]
\addplot[domain=0:1,color=black, very thick,  name path=X] {(5*pi*(x-0.3)+sin(deg(5*pi*(x-0.3))))/12};
\addplot[color=black, name path=Y] coordinates {(0,1) (1,1)};
\addplot[color=black, name path=Z] coordinates {(1,0) (1,1)};
\end{axis}
\end{tikzpicture}

\caption{The thick black curve is a trade boundary $\mathcal{B}$ that is non-decreasing.} \label{fig:S5}
\end{figure}
\indent The main idea of searching for a maxmin deterministic trading mechanism is as follows. I divide all possible deterministic DSIC and EPIR trading mechanisms into four classes according to the trade boundary. By strong duality, I can work on the dual program. I propose a relaxation of the dual program by ignoring a lot of constraints. The merit of doing so is to have a finite-dimensional linear programming problem.  Then I derive an upper bound of the value of the relaxation and  show that it can be attained by constructing  deterministic DSIC and EPIR trading mechanisms as well as a feasible value distribution.
\begin{theorem}\label{t3}
When $\sqrt{M_S}+\sqrt{1-M_B}< 1$ (See Figure \ref{fig:S3}), any deterministic DSIC and EPIR trading mechanism satisfying the following properties is a maxmin deterministic trading mechanism (See Figure \ref{fig:S4}):\\
(i). $(1-\sqrt{1-M_B},0) \in \mathcal{B}, (1,\sqrt{M_S})\in \mathcal{B}$.\\
(ii). $\mathcal{B}$ is above (including) the line  $\sqrt{M_S}v_B-\sqrt{1-M_B}v_S=\sqrt{M_S}(1-\sqrt{1-M_B})$.\\
(iii). The payment rule and the transfer rule are characterized by Proposition \ref{p1}.\\
The profit guarantee is $(1-\sqrt{M_S}-\sqrt{1-M_B})^2$.
The worst-case value distribution puts probability masses of $\sqrt{1-M_B},\sqrt{M_S}$ and $1-\sqrt{1-M_B}-\sqrt{M_S}$ on the value profiles $(1-\sqrt{1-M_B},0), (1,\sqrt{M_S})$ and $(1,0)$ respectively. \\
\indent When $\sqrt{M_S}+\sqrt{1-M_B}\ge 1$, the Never Trading Mechanism is a maxmin deterministic trading mechanism with a profit guarantee of 0. 
\end{theorem}
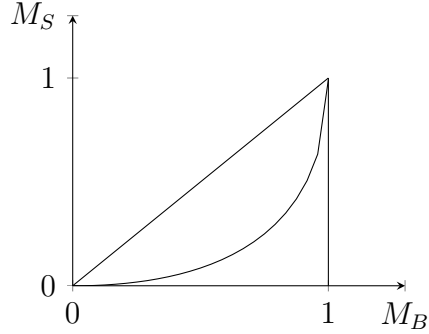
\begin{figure}
\centering
\begin{tikzpicture}
\begin{axis}[
    axis lines = left,
    xmin=0,
        xmax=1.3,
        ymin=0,
        ymax=1.3,
        xtick={0,1,1.3},
        ytick={0,1,1.3},
        xticklabels = {$0$,  $1$, $M_B$},
        yticklabels = {$0$,  $1$, $M_S$},
        legend style={at={(1.1,1)}}
]
\addplot[domain=0:1,color=black, name path=X] {2-x-2*(1-x)^0.5};
\addplot[color=black, name path=Y] coordinates {(0,0) (1,1)};
\addplot[color=black, name path=Z] coordinates {(1,0) (1,1)};
\end{axis}
\end{tikzpicture}

\caption{The black curve is $\sqrt{M_B}+\sqrt{1-M_S}=1$. The revenue guarantee of a maxmin deterministic trading mechanism is positive if and only  $(M_B,M_S)$ lies below  this curve. } \label{fig:S3}
\end{figure}
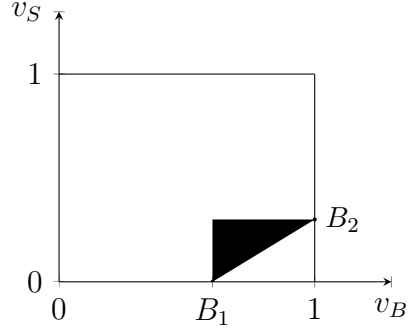
\begin{figure}
\centering
\begin{tikzpicture}
\begin{axis}[
    axis lines = left,
    xmin=0,
        xmax=1.3,
        ymin=0,
        ymax=1.3,
        xtick={0,0.6,1,1.3},
        ytick={0, 1,1.3},
        xticklabels = {$0$, $B_1$, $1$, $v_B$},
        yticklabels = {$0$,  $1$, $v_S$},
        legend style={at={(1.1,1)}}
]
\path[name path=axis] (axis cs:0,0) -- (axis cs:1,0);
\path[name path=A] (axis cs: 0.6,0.3) -- (axis cs:1,0.3);
\path[name path=B] (axis cs:0.6,0) -- (axis cs:1,0.3);
\addplot[area legend, black] fill between[of=A and B,  soft clip={domain=0.6:1}];
\node[black,right] at (axis cs:1,0.3){\small{$B_2$}};
\node at (axis cs:1,0.3) [circle, scale=0.1, draw=black,fill=black] {};
\node at (axis cs:0.6,0) [circle, scale=0.1, draw=black,fill=black] {};
\addplot[color=black, name path=Y] coordinates {(0,1) (1,1)};
\addplot[color=black, name path=Z] coordinates {(1,0) (1,1)};
\end{axis}
\end{tikzpicture}

\caption{$B_1=(\sqrt{1-M_B},0), B_2=(1,\sqrt{M_S})$. If $\sqrt{M_S}+\sqrt{1-M_B}<1$, then  $B_1\in \mathcal{B}, B_2\in \mathcal{B}$, and $\mathcal{B}$ lies in the black region for a maxmin deterministic  trading mechanism.} \label{fig:S4}
\end{figure}
That is, I characterize the  class of maxmin deterministic trading mechanisms for any pair of  expectations with a non-trivial profit guarantee. The worst-case value distribution is discrete, and is the same for the mechanisms in this class. Now I provide examples  of some maxmin deterministic trading mechanisms.\\
\indent \textit{Linear Trading Mechanism}: trade takes place with probability one if  $\sqrt{M_S}v_B-\sqrt{1-M_B}v_S>\sqrt{M_S}(1-\sqrt{1-M_B})$, and conditional on trading, the buyer pays $1-\sqrt{1-M_B}+\sqrt{\frac{1-M_B}{M_S}}v_S$ and the seller receives $\sqrt{\frac{M_S}{1-M_B}}[v_B-(1-\sqrt{1-M_B})]$; otherwise, no trade takes place, and no one pays or receives anything.
\\
\indent \textit{Double Posted-Price  Trading Mechanism}: trade takes place with probability one if  $v_B>1-\sqrt{1-M_B}$ and $v_S< \sqrt{M_S}$, and conditional on trading, the buyer pays $1-\sqrt{1-M_B}$ and the seller receives $\sqrt{M_S}$; otherwise, no trade takes place, and no one pays or receives anything.

\section{Extension and Discussion}\label{s7}
\subsection{Can-hold Case}\label{s71}
Consider a more general model in which the intermediary can hold the asset. To wit, this  only requires that  the   sum of  the buyer's allocation and the seller's allocation  do not exceed 1. Recall that this sum is required to be  1 for the main results.  Formally, the intermediary seeks a trading mechanism $(q_B,q_S,t_B,t_S)$ such that the following  constraints hold:  \[v_Bq_B(v)-t_B(v)\ge v_Bq_B(v_B',v_S)-t_B(v_B',v_S),  \quad \forall  v\in V,v_B'\in V_B;\tag{$DSIC_B$}\]
\[v_Bq_B(v)-t_B(v)\ge 0, \quad \forall v\in V;\tag{$EPIR_B$}\]
\[v_Sq_S(v)+t_S(v)\ge v_Sq_S(v_B,v_S')+t_S(v_B,v_S'),\quad \forall v\in V,v_S'\in V_S ;\tag{$DSIC_S'$}\]
\[v_Sq_S(v)+t_S(v)\ge v_S, \quad \forall v\in V; \tag{$EPIR_S'$}\]\[
 q_B(v)+q_S(v)\le 1,\quad \forall v\in V. \quad\tag{$CH$}\]
 I denote the set of such trading mechanisms as $\mathcal{D}'$\footnote{Note that here the monotonicity constraints are that $q_B(v_B,v_S)$ is non-decreasing w.r.t. $v_B$ for any $v_S$ and $q_S(v_B,v_S)$ is non-decreasing w.r.t. $v_S$ for any $v_B$.}. The intermediary's problem is to seek for a trading mechanism that solves  \[
 \sup_{(q_B,q_S,t_B,t_S)\in \mathcal{D}'}\inf_{\pi\in \Pi(M_B,M_S)}\int_Vt(v)d\pi(v). \tag{MTM'}\label{mtm'}\]
 \begin{theorem}\label{t4}
 For the symmetric (resp, the asymmetric) case, the random double auction (resp, the logarithmic trading mechanism) is a solution to \eqref{mtm'}. 
 \end{theorem}
That is,  the solution to the more general  problem \eqref{mtm'} coincides with the main results. To see this, first note that the the value of \eqref{mtm'} is weakly higher than the value of \eqref{mtm} because $\mathcal{D}\subset \mathcal{D}'$. I will show that the the value of \eqref{mtm} is weakly higher than the value of \eqref{mtm'}. Indeed, given the constructed joint distribution, the constructed trading mechanism is an optimal mechanism even among this wider class of trading mechanism $\mathcal{D}'$. To show this,  first note that a simple adaption of Proposition \ref{p1} yields an analogous virtual representation of the expected profit for this environment: \[E[t(v)]=\int_v[q_B(v)\Phi_B(v)+q_S(v)\Phi_S(v)]dv-1,\]
where $\Phi_B(v)= 
\pi(v)v_B-(\pi_S(v_S)-\Pi_B(v_B,v_S)) $ and $\Phi_S(v)=\pi(v)v_S+\Pi_S(v_S,v_B)$. Here $\Phi_B(v)$ (resp, $\Phi_S(v)$) is the buyer's (resp, the seller's) weighted virtual value when the value profile is $v=(v_B,v_S)$. Given the constructed joint distribution, $\Phi_B=\Phi_S >0$ for any  value profile in the support except for the highest joint type (1,0), in which $\Phi_B(1,0)>\Phi_S (1,0) =0$; in addition, $\Phi_B\le 0$ and $\Phi_S\ge 0$ for any value profile outside the support. This is true for both the symmetric case and the asymmetric case.  Then any  trading mechanism in $\mathcal{D}'$ will be optimal if 1) the ex-post participation constraints are binding for zero-value buyer and one-value seller, and 2) $q_B=0$ and $q_S=1$ for any value profile outside the support, $q_B+q_S=1$ for any value profile inside the support and $q_B(1,0)=1,q_S(1,0)=0$.  It is straightforward to  see that the constructed trading mechanism\footnote{In this more general model, $q_B=q^*$ and $q_S=1-q^*$ (resp, $q_B=q^{**}$ and $q_S=1-q^{**}$ ) for the symmetric case (resp, the asymmetric case).} is such a mechanism and therefore remains optimal to the constructed joint distribution in this general model. Indeed, given the property that the buyer's weighted virtual value is equal to the seller's weighted virtual value for any value profile in the support except for (1,0), the intermediary does not have an incentive to hold the asset in an optimal trading mechanism.
\subsection{Information Design Problem}\label{s72}
A well-known result in models of private information  is that the distribution of agents' private information is a key determinant of their welfare. For example,  in the environment of  bilateral trade, \cite{myerson1983efficient} consider the independent private value model and show that  the two trading parties' welfare  is not the full surplus for general distributions and the amount of their welfare depends on their distributions of valuations. Indeed, most of the existing models of private information in the environment of bilateral trade assume that the distribution of the two trading parties' private information is exogenous. However, it is conceivable that a  \textit{financial regulator}, e.g., the Security and the Exchange Commission (SEC),  may optimally design the nature of  the private information held by the two trading parties to maximize their welfare, given the fact that their welfare is affected by the distribution of their private information.\\
\indent In this section, I consider an information design problem of a \textit{financial regulator}  whose objective is to maximize   \textit{the total expected gain from trade}, defined as the sum of the traders' expected profits. I assume that the financial regulator can carefully design the private information of the traders by choosing a value distribution $\pi\in \Delta(V)$.   The intermediary, after observing the choice of the distribution but not the realized joint valuations, designs a profit-maximizing trading mechanism across DSIC and EPIR trading mechanisms.  Formally, the financial regulator solves \footnote{If the intermediary has multiple optimal trading mechanisms, I break ties in favor of the financial regulator by selecting one that maximizes the gain from trade for the traders. This is a standard tie-breaking rule in the information design literature (e.g., \cite{kamenica2011bayesian}, \cite{roesler2017buyer}) and \cite{condorelli2020information}.} \[
\sup_{\pi\in \Delta(V)}\int[q^*(v)(v_B-v_S)-t^*(v)]d\pi(v) \tag{MGT}\label{fi}\]
subject to \[(q^*,t_B^*,t_S^*)\in \arg\sup_{(q,t_B,t_S)\in \mathcal{D}}\int_vt(v)d\pi(v).\]

\begin{theorem}\label{t6}
The symmetric triangular value distribution with $r=\frac{1}{2}$ is a solution to \eqref{fi}.
\end{theorem}
\indent That is, the constructed symmetric triangular value  distribution with $r=\frac{1}{2}$ is indeed a solution to the financial regulator's information design problem \eqref{fi}.\\
\indent Recall that a symmetric triangular value distribution has the property that the weighted virtual value is zero  for any value profile in the support except for the value profile $(1,0)$. This property has two implications. First, it implies that  an \textit{efficient} \footnote{Precisely, trade takes place with probability one for any report in the support of the symmetric triangular value distribution, and 0 otherwise.} trading mechanism is a best response for the intermediary. Second, it implies that in a best response,   the intermediary does not discriminate across all value profiles in the support but the value profile $(1,0)$. These two implications render a symmetric triangular value distribution a good candidate as a solution. Under a symmetric triangular value distribution, the total expected gain from trade  is  the difference between the first-best   gain from trade and the expected profit of the intermediary.  If $r$ is high, then the first-best  gain from trade is high, but the expected profit of the intermediary is  high as well. In the extreme case where $r=1$, the symmetric triangular value distribution is reduced to a point mass on the value profile $(1,0)$, under which the first-best gain from trade and the expected profit of the intermediary both are one, resulting in a zero total expected gain from trade. The optimal $r=\frac{1}{2}$  maximizes the difference between the first-best gain from trade and the expected profit of the intermediary.\\
\indent Under the symmetric triangular value distribution with $r=\frac{1}{2}$,  the first-best gain from trade is $\frac{3}{4}$ and it is shared equally among the buyer, the seller and the intermediary, so that each obtains an expected profit of $\frac{1}{4}$. The total expected gain from trade is $\frac{1}{2}$.\\
\indent To show that this distribution  is a solution to the problem \eqref{fi}, I construct two equivalent auxiliary problems whose values are weakly higher than that of  the problem \eqref{fi}. Then I show that the values of the auxiliary problems are $\frac{1}{2}$ by constructing a saddle point. Details are relegated to Appendix \ref{ad}.\\ 
\indent The information design problem \eqref{fi} is closely related to \cite{condorelli2020information} who  consider a  \textit{buyer-optimal} information design problem:  the buyer can choose the probability distribution
of her valuation for the good to maximize her profit. The seller, after observing the buyer’s choice of the distribution
but not the realized valuation, designs a revenue-maximizing selling mechanism.  The problem \eqref{fi}  may be interpreted as a \textit{traders-optimal} information design problem.  Critically, trade is efficient under the solution in either problem. 
\section{Concluding Remarks}\label{s8}
In this paper, I construct a \textit{novel} trading mechanism, a random double auction, and show that it maximizes the profit guarantee for a profit-maximizing intermediary who only knows the expectations of the traders' values for the symmetric case. I extend my analysis to characterizing a  maxmin trading mechanism  for the asymmetric case. The key step for the results is to construct a saddle point. The construction method may be of independent interest and  useful for other design problems, e.g., multidimensional 
Bayesian persuasion, and  other robust optimization problems.  To my knowledge, this paper is the first to provide a complete  maxmin solution across all  dominant-strategy incentive compatible and ex-post individually rational mechanisms.  In addition, the profit guarantee is positive for any non-trivial pair of  expectations  in which  the expectation of the buyer's value exceeds the expectation of the seller's value. In contrast, the profit guarantee of a maxmin deterministic trading mechanism is 0 for a positive measure of non-trivial pairs of  expectations. 
\appendix
\section{Proofs for Section \ref{s4}} \label{aa}
\subsection{Proof of Proposition \ref{p1}}
(i). Dominant-strategy incentive compatibility (DSIC) for a type $v_B$ of $B$ requires that for any $v_S$ and $v_B'\neq v_B$: $$v_Bq(v_B,v_S) - t_B(v_B,v_S)\ge v_Bq(v_B',v_S) - t_B(v_B',v_S). $$DSIC for a type $v_B'$ of $B$  requires that for any $v_S$ and $v_B\neq v_B'$:$$ v_B'q(v_B',v_S) - t_B(v_B',v_S)\ge v_B'q(v_B,v_S) - t_B(v_B,v_S).$$Adding the two inequalities, I have that: $$(v_B-v_B')(q(v_B,v_S)-q(v_B',v_S))\ge 0.$$It follows that $q(v_B,v_S)\ge q(v_B',v_S)$ whenever $v_B>v_B'$ .\\
Similarly, DSIC for a type $v_S$ of $S$ requires that for any $v_B$ and $v_S'\neq v_S$: $$v_S(1-q(v_B,v_S)) + t_B(v_B,v_S)\ge v_S(1-q(v_B,v_S')) + t_S(v_B,v_S').  $$DSIC for a type $v_S'$ of $S$ requires that for any $v_B$ and $v_S\neq v_S'$:$$ v_S'(1-q(v_B,v_S')) +t_S(v_B,v_S')\ge v_S'(1-q(v_B,v_S)) + t_S(v_B,v_S).$$Adding the two inequalities, I have that: $$(v_S-v_S')(q(v_B,v_S')-q(v_B,v_S))\ge 0.$$It follows that $q(v_B,v_S)\le q(v_B,v_S')$ whenever $v_S>v_S'$ .\\\\
(ii). Fix $v_S$, and define $$U_B(v_B):=v_Bq(v_B,v_S)-t_B(v_B,v_S).$$ 
By the first two inequalities in (i), I obtain that $$(v_B'-v_B)q(v_B,v_S)\le U_B(v_B')-U_B(v_B)\le (v_B'-v_B)q(v_B',v_S).$$
Therefore $U_B(v_B)$ is Lipschitz, hence absolutely continuous w.r.t. $v_B$ and therefore differentiable w.r.t. $v_B$ almost everywhere. Then applying the envelope theorem to the above inequality at each point of differentiability, I obtain that $$
\frac{dU_B(v_B)}{dv_B}=q(v_B,v_S).$$Then I have that $$t_B(v_B,v_S)= v_Bq(v_B,v_S)-\int_{0}^{v_B}q(x,v_S)dx-U_B(0).$$
Note that $U_B(0)\ge 0$ by the EPIR constraint. If $U_B(0)>0$, then I can reduce it to 0 so that I can increase the payment from $B$ for  any value profile in which the seller's value is $v_S$.  And the profit guarantee  will be weakly greater. Thus, when searching for a maxmin trading mechanism, it is without loss of generality to let $U_B(0)=0$. Then I obtain that  $t_B(v_B,v_S)=v_Bq(v_B,v_S)-\int_{0}^{v_B}q(x,v_S)dx$.\\
(iii).  Similarly, fix $v_B$, and define $$U_S(v_S):=v_S(1-q(v_B,v_S))+t_S(v_B,v_S).$$ 
By the fourth and fifth inequalities in (i), I obtain that $$(v_S'-v_S)(1-q(v_B,v_S))\le U_S(v_S')-U_S(v_S)\le (v_S'-v_S)(1-q(v_B,v_S')).$$ Therefore $U_S(v_S)$ is Lipschitz, hence absolutely continuous w.r.t. $v_S$ and therefore differentiable w.r.t. $v_S$ almost everywhere. Then applying the envelope theorem to the above inequality at each point of differentiability, I obtain that $$
\frac{dU_S(v_S)}{dv_S}=1-q(v_B,v_S).$$Then I have that $$t_B(v_B,v_S)= U_S(1)-v_S(1-q(v_B,v_S))-\int_{v_S}^{1}q(v_B,x)dx.$$
Note that $U_S(1)\ge 1$ by the EPIR constraint. If $U_S(1)>1$, then I can reduce it to 1 so that I can decrease the payment to $S$ for any value profile in which the buyer's value is $v_B$.  And the profit guarantee will be weakly greater. Thus, when searching for a maxmin trading mechanism, it is without loss of generality to let $U_S(1)=1$. Then I obtain that  $t_S(v)=1-(1-q(v))v_S-\int_{v_S}^1(1-q(v_B,x)dx=q(v)v_S+\int_{v_S}^1q(v_B,x)dx$.\\
(iv). This is implied by (ii) and (iii).
\subsection{  Proof of Lemma \ref{l1}} 
Given a DSIC and EPIR trading mechanism $(q,t_B,t_S)$, the primal minimization problem of adversarial nature is as follows (with dual variables in the bracket): 
 \[\inf_{\pi}\int t(v)d\pi(v)\tag{P}\label{p}\] subject to 
    \[\int v_Bd\pi(v)=M_B,\quad (\lambda_B)\]
    \[\int v_Sd\pi(v)=M_S,\quad (\lambda_S)\]
    \[\int d\pi(v)=1. \quad (\mu)\]
It has the following  dual maximization problem:
    \[\sup_{\lambda_B,\lambda_S,\mu \in \mathcal{R}}\lambda_B M_B+\lambda_S M_S+\mu\tag{D}\label{D}\] subject to
    \[\lambda_Bv_B+\lambda_Sv_S+\mu \le t(v),\quad\forall v\in V.\]
Note that the value of \eqref{p} is bounded by 1 as $t(v)\le 1$. In addition, the  joint distribution that puts all probability masses on the value profile $(M_B,M_S)$  is in the interior of the primal cone. Then by Theorem 3.12 in \cite{anderson1987linear}, strong duality holds. Then, by the complementary slackness, \eqref{cs} holds. 
\section{Proofs for Section \ref{s5}}\label{ab}
\subsection{ Proof of Theorem \ref{t1}}
\textit{Step 1}: The random double auction maximizes the expected profit under the symmetric triangular value distribution.  To show this, first note that  \eqref{eq16} holds by construction. In addition,  the weighted virtual value  is non-positive for any value profile outside the support and positive for the value profile $(1,0)$.   Then any  DSIC and EPIR trading mechanism in which 1) ex-post participation constraints are binding for zero-value buyer and one-value seller, and 2) trade does not take place  when $v_B - v_S< r$ and trade takes place with probability one when $(v_B,v_S) =  (1,0)$ is an optimal trading mechanism. It is straightforward to see that the random double auction is such a mechanism.\\
\textit{Step 2}:  The symmetric triangular value distribution minimizes the expected profit under the random double auction. I use the duality theory to show this. Note that the symmetric triangular value distribution is a feasible value distribution by construction.  By the weighted virtual value representation, the value of \eqref{p} given the random double auction and the symmetric triangular value distribution is $Pr(1,0)\times (1-0)=r^2$. Second,  the constraints in \eqref{D} hold for all value profiles. To see this,  note  that for any value profile $v=(v_B,v_S)$  in the support (or $v\in ST$),   $\lambda_B^* v_B+\lambda_S^* v_S + \mu^*  =t^*(v)$ by \eqref{eq3}. Also,  for any value profile $v=(v_B,v_S)$  in which $v_B-v_S=r$,  $\lambda_B^* v_B+\lambda_S^* v_S + \mu^*  =0=t^*(v)$ because $\lambda_B^*=\frac{r}{1-r}$, $\lambda_S^*=-\frac{r}{1-r}$ and $\mu^*=-\frac{r^2}{1-r}$. Then,  for any value profile $v=(v_B,v_S)$ in which $v_B-v_S<r$, $\lambda_B^* v_B+\lambda_S^* v_S + \mu^* < 0=t^*(v)$ because $\lambda_B^*>0$ and $\lambda_S^*<0$.  Finally,  the value of \eqref{D} given the constructed dual variables is equal to $r^2$ by simple calculation.
\subsection{Proof of Lemma \ref{l2}}
I start from establishing the following four claims regarding some properties of the functions $H_1(r_1,r_2)$ and $H_2(r_1,r_2)$, which play a crucial role in establishing Lemma \ref{l2}. First, by simple calculation, I have that for $(r_1,r_2)\in (0,1)^2$, 
\[\label{eq57}
    H_1(r_1,r_2)=\frac{(1-r_2)r_1(1-r_1)^2}{(1-r_1-r_2)^2}\ln{\frac{1-r_2}{r_1}}-\frac{r_1r_2(1-r_1)}{1-r_1-r_2}+r_1,\tag{B.2.1}
\]
\[\label{eq42}
    H_2(r_1,r_2)=\frac{r_1(1-r_2)r_2^2}{(1-r_1-r_2)^2}\ln{\frac{1-r_2}{r_1}}-\frac{r_1r_2^2}{1-r_1-r_2}.\tag{B.2.2}
\]
 First note that $H_1$ and $H_2$ are not well-defined when $0<r_1=1-r_2<1$. Using   L'H\^{o}pital's rule, it is straightforward to show that  $\lim_{1-r_2\to r_1}H_1(r_1,r_2)=\frac{1-r_1^2+2r_1}{2}$ and  $\lim_{1-r_2\to r_1}H_2(r_1,r_2)=\frac{(1-r_1)^2}{2}$ . I thus define $H_1(r_1,r_2):=\lim_{1-r_2\to r_1}H_1(r_1,r_2)$ and  $H_2(r_1,r_2):=\lim_{1-r_2\to r_1}H_2(r_1,r_2)$ when $0<r_1=1-r_2<1$. This makes  $H_1$ and $H_2$ continuous on $(0,1)^2$. In addition, using L'H\^{o}pital's rule,  it is straightforward to show that $\lim_{r_1\to 0} H_1(r_1,r_2)=0$ for $r_2\in (0,1)$, $\lim_{r_1\to 1}H_1(r_1,r_2)=1$ for $r_2\in (0,1)$, $\lim_{r_2\to 0}H_1(r_1,r_2)=r_1-r_1\ln{r_1}$ for $r_1\in (0,1)$, $\lim_{r_2\to 1}H_1(r_1,r_2)=1$ for $r_1\in (0,1)$, $\lim_{r_1\to 0}H_2(r_1,r_2)=0$ for $r_2\in (0,1)$,  $\lim_{r_1\to 1}H_2(r_1,r_2)=(1-r_2)\ln{(1-r_2)}+r_2$ for $r_2\in (0,1)$, $\lim_{r_2\to 0}H_2(r_1,r_2)=0$ for $r_1\in (0,1)$ and $\lim_{r_2\to 1} H_1(r_1,r_2)=1$ for $r_1\in (0,1)$. Therefore I define $H_1(r_1,r_2)$ and $H_2(r_1,r_2)$ as follows. \[
H_1(r_1,r_2)=
\left\{
\begin{array}{lll}
\frac{(1-r_2)r_1(1-r_1)^2}{(1-r_1-r_2)^2}\ln{\frac{1-r_2}{r_1}}-\frac{r_1r_2(1-r_1)}{1-r_1-r_2}+r_1     &      & { \text{if $(r_1,r_2)\in (0,1)^2$ and $r_1+r_2\neq 1$},}\\
\frac{1-r_1^2+2r_1}{2}     &      & {\text{if}\quad 0<r_1=1-r_2<1,}\\
0      &      & {\text{if $r_1=0$ an  $r_2\in (0,1)$},}
\\
1      &      & {\text{if $r_1=1$ an  $r_2\in (0,1)$},}
\\
r_1-r_1\ln{r_1}      &      & {\text{if $r_2=0$ an  $r_1\in (0,1)$},}
\\
1      &      & {\text{if $r_2=1$ an  $r_1\in (0,1)$}.}
\end{array} \right.\]
\[
H_2(r_1,r_2)=
\left\{
\begin{array}{lll}
\frac{r_1(1-r_2)r_2^2}{(1-r_1-r_2)^2}\ln{\frac{1-r_2}{r_1}}-\frac{r_1r_2^2}{1-r_1-r_2}    &      & { \text{if $(r_1,r_2)\in (0,1)^2$ and $r_1+r_2\neq 1$},}\\
\frac{(1-r_1)^2}{2}    &      & {\text{if}\quad 0<r_1=1-r_2<1,}\\
0      &      & {\text{if $r_1=0$ an  $r_2\in (0,1)$},}
\\
(1-r_2)\ln{(1-r_2)}+r_2      &      & {\text{if $r_1=1$ an  $r_2\in (0,1)$},}
\\
0     &      & {\text{if $r_2=0$ an  $r_1\in (0,1)$},}
\\
1      &      & {\text{if $r_2=1$ an  $r_2\in (0,1)$}.}
\end{array} \right.\]

\begin{claim}\label{cl3}
Fix any $r_2\in [0,1)$, $H_1(r_1,r_2)$ is strictly increasing in $r_1$. Moreover, for any $r_2\in (0,1)$,  $\lim_{r_1 \to 1-r_2}\frac{\partial H_1(r_1,r_2)}{\partial r_1}$ exists and is positive.   In addition, for any $r_2\in [0,1)$, as $r_1\to 1$, $H_1(r_1,r_2)\rightarrow 1$.
\end{claim}
\begin{proof}[Proof of Claim \ref{cl3}]
When $r_2=0$, $H_1(r_1,r_2)=r_1-r_1\ln{r_1}$. This is an strictly increasing function because  $\frac{\partial H_1(r_1,r_2)}{\partial r_1}=-\ln{r_1}>0$. In addition,  by  L'H\^{o}pital's rule, $\lim_{r_1\to 1}H_1(r_1,r_2)=1$. Thus, Claim \ref{cl3} holds when $r_2=0$. When $0<r_2<1$, I already have that $\lim_{r_1\to 1}H_1(r_1,r_2)=1$. Now 
taking the first order derivative w.r.t. $r_1$ to \eqref{eq57},  I obtain that 
\[\label{eq58}
    \frac{\partial H_1(r_1,r_2)}{\partial r_1}=\frac{(1-r_1)(1-r_2)}{(1-r_1-r_2)^2}[(1-3r_1+\frac{2r_1(1-r_1)}{1-r_1-r_2})\ln{\frac{1-r_2}{r_1}}-2r_2].\tag{B.2.3}
\]
Then to show the first part of Claim \ref{cl3}, it suffices to show that if $(r_1,r_2) \in (0,1)^2$ and $r_1+r_2\neq 1$,
\[\label{eq59}
    (1-3r_1+\frac{2r_1(1-r_1)}{1-r_1-r_2})\ln{\frac{1-r_2}{r_1}}-2r_2>0. \tag{B.2.4}
\]
Let $\beta:=\frac{1-r_2}{r_1}$, then $\beta\in (0,1)\cup (1,\infty)$. Plugging $r_2=1-\beta r_1$ into \eqref{eq59}, it suffices to show that for any $\beta \in (0,1)\cup (1,\infty)$, \[\label{eq60}
    (1-3r_1+\frac{2(1-r_1)}{\beta-1})\ln{\beta}-2(1-\beta r_1)>0. \tag{B.2.5}
\]
Slightly rewriting \eqref{eq60}, it suffices to show that for any $\beta\in (0,1)\cup (1,\infty)$,
\[\label{eq61}
    \frac{\beta+1}{\beta-1}\ln{\beta}-2 +(-\frac{3\beta-1}{\beta-1}\ln{\beta}+2\beta)r_1>0. \tag{B.2.6}
\]
Then, it suffices to show that for  any $\beta\in (0,1)\cup (1,\infty)$, the following two inequalities hold:
\[\label{eq62}
    \frac{\beta+1}{\beta-1}\ln{\beta}-2>0, \tag{B.2.7}
\]
\[\label{eq63}
    -\frac{3\beta-1}{\beta-1}\ln{\beta}+2\beta>0. \tag{B.2.8}
\]
Now to prove \eqref{eq62}, it suffices to show that $f(\beta):=\ln{\beta}-\frac{2(\beta-1)}{\beta+1}>0$ for $\beta\in (1,\infty)$
and $f(\beta)<0$ for $\beta\in (0,1)$. Taking the first order derivative to $f(\beta)$, I obtain that 
\[\label{eq49}
    f'(\beta)=\frac{(\beta-1)^2}{\beta(\beta+1)^2}.\tag{B.2.9}
\]
Therefore, $f(\beta)$ is strictly increasing. Note that $f(1)=0$. Thus, I proved \eqref{eq62}.  To prove \eqref{eq63}, it suffices to show that   $h(\beta):=(1-3\beta)\ln{\beta}+2\beta(\beta-1)>0$ for $\beta\in (1,\infty)$
and $h(\beta)<0$ for $\beta\in (0,1)$. Taking the first order derivative to $h(\beta)$, I obtain that 
\[\label{eq64}
    h'(\beta)=4\beta-3\ln{\beta}+\frac{1}{\beta}-5. \tag{B.2.10}
\]
Now taking the second order derivative to  $h(\beta)$, I obtain that 
\[\label{eq65}
    h''(\beta)=\frac{(4\beta+1)(\beta-1)}{\beta^2}. \tag{B.2.11}
\]
Note that $h''(\beta)>0$ when $\beta>1$, $h''(\beta)<0$ when $\beta<1$ and $h''(1)=0$. This implies that $h'(\beta)$ is minimized at $\beta=1$. Note that $h'(1)=0$. This implies that $h(\beta)$ is strictly increasing. Finally, note that $h(1)=0$. This implies that \eqref{eq63} holds. \\
\indent Using L'H\^{o}pital's rule, I have that $\lim_{r_1 \to 1-r_2}\frac{\partial H_1(r_1,r_2)}{\partial r_1}=\frac{r_2(6-5r_2)}{1-r_2}>0$ for $r_2\in (0,1)$. 
\end{proof}
\begin{claim}\label{cl4}
Fix any $r_1\in (0,1) $, $H_1(r_1,r_2)$ is strictly increasing in $r_2$. Moreover, for any $r_1\in (0,1)$, $\lim_{r_2 \to 1-r_1}\frac{\partial H_1(r_1,r_2)}{\partial r_2}$ exists and is positive.  In addition, for any $r_1\in (0,1) $, as $r_2 \to 1$, $H_1(r_1,r_2)\rightarrow 1$.

\end{claim}
\begin{proof}[Proof of Claim \ref{cl4}]
When $0<r_1<1 $, I already have that $\lim_{r_2\to 1}H_1(r_1,r_2)=1$. Now taking the first order derivative w.r.t. $r_2$ to \eqref{eq57}, with some algebra,  I obtain that 
\[\label{eq66}
    \frac{\partial H_1(r_1,r_2)}{\partial r_2}=\frac{(1-r_1)^2r_1}{(1-r_1-r_2)^2}[(-1+\frac{2(1-r_2)}{1-r_1-r_2})\ln{\frac{1-r_2}{r_1}}-2]. \tag{B.2.12}
\]
Then it suffices to show that if $(r_1,r_2) \in (0,1)^2$ and $r_1+r_2\neq 1$,
\[\label{eq67}
    (-1+\frac{2(1-r_2)}{1-r_1-r_2})\ln{\frac{1-r_2}{r_1}}-2>0.\tag{B.2.13}
\]
Plugging $r_2=1-\beta r_1$ into \eqref{eq67}, it suffices to show that for any $\beta\in (0,1)\cup (1,\infty)$,
\[\label{eq68}
    \frac{\beta+1}{\beta-1}\ln{\beta}-2>0. \tag{B.2.14}
\]
This is exactly \eqref{eq62} and has been shown in the Proof of Claim \ref{cl3}.\\
\indent Using L'H\^{o}pital's rule, I have that $\lim_{r_2 \to 1-r_1}\frac{\partial H_1(r_1,r_2)}{\partial r_2}=\frac{(1-r_1)^2}{6r_1}>0$ for $r_1\in (0,1)$. 
\end{proof}
\begin{claim}\label{cl2}
Fix any $r_2\in (0,1)$, $H_2(r_1,r_2)$ is strictly increasing in $r_1$. Moreover, for $r_2\in (0,1)$,  $\lim_{r_1 \to 1-r_2}\frac{\partial H_2(r_1,r_2)}{\partial r_1}$ exists and is positive. 
\end{claim}
\begin{proof}[Proof of Claim \ref{cl2}]
Taking the first order derivative w.r.t. $r_1$ to \eqref{eq42},  I obtain that 
\[\label{eq54}
    \frac{\partial H_2(r_1,r_2)}{\partial r_1}=\frac{(1-r_2)r_2^2}{(1-r_1-r_2)^2}[(1+\frac{2r_1}{1-r_1-r_2})\ln{\frac{1-r_2}{r_1}}-2]. \tag{B.2.15}
\]
Then it suffices to show that if $(r_1,r_2) \in (0,1)^2$ and $r_1+r_2\neq 1$,
\[\label{eq55}
    (1+\frac{2r_1}{1-r_1-r_2})\ln{\frac{1-r_2}{r_1}}-2>0.\tag{B.2.16}
\]
Plugging $r_2=1-\beta r_1$ into \eqref{eq55}, it suffices to show that for any $\beta\in (0,1)\cup (1,\infty)$,
$ \frac{\beta+1}{\beta-1}\ln{\beta}-2>0$, 
which is exactly \eqref{eq62} and has been shown in the Proof of Claim \ref{cl3}. \\
\indent Using L'H\^{o}pital's rule, I have that $\lim_{r_1 \to 1-r_2}\frac{\partial H_2(r_1,r_2)}{\partial r_1}=\frac{(r_2)^2}{6(1-r_2)}>0$ for $r_2\in (0,1)$. 
\end{proof}
\begin{claim}\label{cl1}
Fix any $r_1\in (0,1]$, $H_2(r_1,r_2)$ is strictly increasing in $r_2$. Moreover, for any $r_1\in (0,1)$, $\lim_{r_2 \to 1-r_1}\frac{\partial H_2(r_1,r_2)}{\partial r_2}$ exists and is positive.  In addition, for any $r_1\in (0,1]$, as $r_2 \to 1$, $H_2(r_1,r_2)\rightarrow 1$.
\end{claim}
\begin{proof}[Proof of Claim \ref{cl1}]
When $r_1=1$, $H_2(r_1,r_2)=(1-r_2)\ln{(1-r_2)}+r_2$. This is an strictly increasing function because  $\frac{\partial H_2(r_1,r_2)}{\partial r_2}=-\ln{(1-r_2)}>0$. In addition,  by  L'H\^{o}pital's rule, $\lim_{r_2\to 1}H_2(r_1,r_2)=1$. Thus, Claim \ref{cl1} holds when $r_1=1$. When $0<r_1<1$, I already have that $\lim_{r_2\to 1}H_1(r_1,r_2)=1$. Now taking the first order derivative w.r.t. $r_2$ to \eqref{eq42},   I obtain that 
\[\label{eq43}
    \frac{\partial H_2(r_1,r_2)}{\partial r_2}=\frac{r_1r_2}{(1-r_1-r_2)^2}[(2-3r_2+\frac{2r_2(1-r_2)}{1-r_1-r_2})\ln{\frac{1-r_2}{r_1}}-2(1-r_1)].\tag{B.2.17}
\]
Then to show the first part of Claim \ref{cl1}, it suffices to show that if $(r_1,r_2) \in (0,1)^2$ and $r_1+r_2\neq 1$,
\[\label{eq44}
    (2-3r_2+\frac{2r_2(1-r_2)}{1-r_1-r_2})\ln{\frac{1-r_2}{r_1}}-2(1-r_1)>0. \tag{B.2.18}
\]
 Plugging $r_2=1-\beta r_1$ into \eqref{eq44}, it suffices to show that for any $\beta\in (0,1)\cup (1,\infty)$,
\[\label{eq45}
    (3\beta r_1-1+\frac{2\beta(1-\beta r_1)}{\beta-1})\ln{\beta}-2(1-r_1)>0.\tag{B.2.19}
\]
Slightly rewriting \eqref{eq45}, it suffices to show that for any $\beta \in (0,1)\cup (1,\infty)$,
\[ \label{eq46}
    \frac{\beta+1}{\beta-1}\ln{\beta}-2 +(\frac{\beta^2-3\beta}{\beta-1}\ln{\beta}+2)r_1>0.\tag{B.2.20}
\]
Then, it suffices to show that for  any $\beta \in (0,1)\cup (1,\infty)$, the following two inequalities hold:
\[ \label{eq47}
    \frac{\beta+1}{\beta-1}\ln{\beta}-2>0,\tag{B.2.21}
\]
\[\label{eq48}
    \frac{\beta^2-3\beta}{\beta-1}\ln{\beta}+2>0.\tag{B.2.22}
\]
Note that \eqref{eq47} is exactly  \eqref{eq62}, which has been shown in the Proof of Claim \ref{cl3}.  To prove \eqref{eq48}, it suffices to show that  $g(\beta):=(\beta^2-3\beta)\ln{\beta}+2(\beta-1)>0$ for $\beta\in (1,\infty)$
and $g(\beta)<0$ for $\beta\in (0,1)$. Taking the first order derivative to $g(\beta)$, I obtain that 
\[\label{eq50}
    g'(\beta)=(2\beta-3)\ln{\beta}+\beta-1. \tag{B.2.23}
\]
Now taking the second order derivative to $g(\beta)$, I obtain that 
\[\label{eq51}
    g''(\beta)=2\ln{\beta}-\frac{3}{\beta}+3.\tag{B.2.24}
\]
Note that $g''(\beta)$ is strictly increasing and $g''(1)=0$. This implies that $g'(\beta)$ is minimized at $\beta=1$. Note that $g'(1)=0$. This implies that $g(\beta)$ is strictly increasing. Finally, note that $g(1)=0$. This implies that \eqref{eq48} holds. \\
\indent Using L'H\^{o}pital's rule, I have that $\lim_{r_2 \to 1-r_1}\frac{\partial H_2(r_1,r_2)}{\partial r_2}=\frac{(1-r_1)(5r_1+1)}{6r_1}>0$ for $r_1\in (0,1)$. 
\end{proof}

I am now ready to prove Lemma \ref{l2}. Fix any $(M_B,M_S)$ in which $0<M_S<M_B<1$ and $M_B+M_S\neq 1$. By Claim \ref{cl3}, \ref{cl4} and the Intermediate Value Theorem, I have that  for any $r_2\in [0,1)$, there exists a unique  $I(r_2)\in (0,1)$ such that $r_1=I(r_2)$ is a solution to $H_1(r_1,r_2)=M_B$. In addition, $I(r_2)$ is a strictly decreasing function. Moreover, by the Implicit Function Theorem\footnote{The Implicit Function Theorem applies for any $r_2\in [0,1)$ because by Claim \ref{cl3} and \ref{cl4}, $\frac{\partial H_1(I(r_2),r_2)}{\partial r_1}>0$ and $\frac{\partial H_1(I(r_2),r_2)}{\partial r_2}>0$ for any $r_2\in [0,1)$.  }, $I(r_2)$ is continuous at each $r_2\in [0,1)$.   By Claim \ref{cl2}, \ref{cl1} and the Intermediate Value Theorem, I have that for any $r_1\in (0,1]$, there exists a  unique $J(r_1)\in (0,1)$ such that $r_2=J(r_1)$ is a solution to $H_2(r_1,r_2)=M_S$.  In addition, $J(r_1)$ is a strictly decreasing function. Moreover, by the Implicit Function Theorem\footnote{The Implicit Function Theorem applies for any $r_1\in (0,1]$ because by Claim \ref{cl2} and \ref{cl1}, $\frac{\partial H_2(r_1,J(r_1))}{\partial r_1}>0$ and $\frac{\partial H_2(r_1,J(r_1))}{\partial r_2}>0$ for any $r_1\in (0,1]$.  }, $J(r_1)$ is continuous at each $r_1\in (0,1]$. Thus it suffices to prove that there exists $r_2\in (0,1)$ such that 
\[\label{eq70}
    J(I(r_2))=r_2. \tag{B.2.25}
\]
Note that $J(I(r_2))$ is a continuous and strictly increasing function for $r_2\in [0,1)$. Also note that $J(I(0))\in (0,1)$ because $I(0) \in (0,1)$ and $J(r_1) \in (0,1)$ when $r_1 \in (0,1)$. Now, by the Intermediate Value Theorem,  it suffices to show that there exists some $r_2\in (0,1)$ such that 
\[\label{eq71}
    J(I(r_2))\le r_2. \tag{B.2.26}
\]
Because $J$ is strictly decreasing, it is equivalent to showing that  there exists some $r_2\in (0,1)$ such that 
\[\label{eq72}
    I(r_2)\ge J^{-1}(r_2). \tag{B.2.27}
\]
By Claim \ref{cl3}, this is equivalent to showing that there exists some $r_2\in (0,1)$ such that 
\[\label{eq73}
    H_1(J^{-1}(r_2),r_2)\le M_B. \tag{B.2.28}
\]
Let $\epsilon:= M_B-M_S >0$. I observe a relationship between the two functions $H_1$ and $H_2$ when $(r_1,r_2)\in (0,1)^2$: 
\[\label{eq74}
    H_1(r_1,r_2)-H_2(r_1,r_2)= (\frac{(1-r_1)^2}{r_2^2}-1)H_2(r_1,r_2)+r_1(2-r_1). \tag{B.2.29}
\]
Note that when $r_2\to 1$, $J^{-1}(r_2) \rightarrow 0$. To see this, suppose not, then by Claim \ref{cl1}, $H_2(J^{-1}(r_2),r_2)\rightarrow 1$ when $r_2\to 1$, a contradiction to $H_2(J^{-1}(r_2),r_2)= M_S<1$. Then by \eqref{eq74}, as $r_2\to 1$, 
\[
\begin{split}
H_1(J^{-1}(r_2),r_2)-M_S & = H_1(J^{-1}(r_2),r_2)-H_2(J^{-1}(r_2),r_2) \\
 & = (\frac{(1-J^{-1}(r_2))^2}{(r_2)^2}-1)H_2(J^{-1}(r_2),r_2)+J^{-1}(r_2)(2-J^{-1}(r_2))\\
 & =  (\frac{(1-J^{-1}(r_2))^2}{(r_2)^2}-1)M_S+J^{-1}(r_2)(2-J^{-1}(r_2))\\
 & \rightarrow (\frac{(1-0)^2}{1^2}-1)M_S+0(2-0)\\
  & = 0.
\end{split}
\]
This implies that  there exists some $r_2\in (0,1)$ such that 
\[\label{eq75}
    |H_1(J^{-1}(r_2),r_2)-M_S|\le \frac{\epsilon}{2}.\tag{B.2.30}
\]
Note that \eqref{eq75} implies \eqref{eq73} because $H_1(J^{-1}(r_2),r_2)\le M_S+\frac{\epsilon}{2}<M_S+\epsilon=M_B$. 
\indent Finally,  suppose that  $r_1+r_2=1$ for the solution $(r_1,r_2)$, then $M_B+M_S=H_1(r_1,r_2)+H_2(r_1,r_2)=1$ by the definition of $H_1(r_1,r_2)$ and $H_2(r_1,r_2)$, a contradiction to the assumption that $M_B+M_S\neq 1$. Therefore, I have that $r_1+r_2\neq 1$ for the solution $(r_1,r_2)$. 
\subsection{Proof of Theorem \ref{t2}}
\textit{Step 1}: The logarithmic trading mechanism maximizes the expected profit under the asymmetric triangular value distribution. To show this, first note that  \eqref{zw2} holds by construction. In addition,  the weighted virtual value  is non-positive for any value profile outside the support and positive for the value profile $(1,0)$.   Then any  DSIC and EPIR trading mechanism in which 1) ex-post participation constraints are binding for zero-value buyer and one-value seller, and 2) trade does not take place when $r_2v_B - (1-r_1) v_S<r_1r_2$ and trade takes place with probability one when $(v_B,v_S) =  (1,0)$ is an optimal trading mechanism. It is straightforward to see that the logarithmic trading mechanism is such a mechanism.\\
\textit{Step 2}:  The asymmetric triangular value distribution minimizes the expected profit under the logarithmic trading mechanism. I use the duality theory to show this. Note that the asymmetric triangular value distribution is a feasible value distribution by construction.  By the weighted virtual value representation, the value of \eqref{p} given the random double auction and the symmetric triangular value distribution is $Pr(1,0)\times (1-0)=r_1(1-r_2)$. Second,  the constraints in \eqref{D} hold for all value profiles. To see this,  note  that for any value profile $v=(v_B,v_S)$  inside the support (or $v\in AT$),  $\lambda_B^
{**}v_B+\lambda_S^{**} v_S + \mu^{**}  = t^{**}(v)$ by \eqref{eq29}. Also,  for any value profile $v=(v_B,v_S)$ in which $r_2v_B-(1-r_1)v_S=r_1r_2$,  $\lambda_B^
{**}v_B+\lambda_S^{**} v_S + \mu^{**}  = 0=t^{**}(v)$ because $\lambda_B^
{**}=\frac{1-r_1-r_2}{(1-r_1)\ln{\frac{1-r_2}{r_1}}}$, $\lambda_S^
{**}=-\frac{1-r_1-r_2}{r_2\ln{\frac{1-r_2}{r_1}}}$ and $\mu^
{**}=-\frac{r_1(1-r_1-r_2)}{(1-r_1)\ln{\frac{1-r_2}{r_1}}}$.  Then,  for any value profile $v=(v_B,v_S)$  in which $r_2v_B-(1-r_1)v_S<r_1r_2$, $\lambda_B^{**} v_B+\lambda_S^{**} v_S + \mu^{**} < 0=t^{**}(v)$ because $\lambda_B^{**}>0$ and $\lambda_S^{**}<0$.  Finally,  the value of \eqref{D} given the constructed dual variables is equal to $r_1(1-r_2)$ by simple calculation. The details of the constructed dual variables as well as the characterization are given in Appendix \ref{ae}.
\section{Proof for Section \ref{s6}: Theorem \ref{t3}}\label{ac}
\textbf{\textit{Step 1: Narrow down the search to a class of mechanisms.}}\\
\indent I divide all deterministic DSIC and EPIR trading mechanisms that satisfy the properties stated in Proposition \ref{p1} into the following four classes:\\
\textit{Class 1}: the trade boundary touches on the value profiles $(c_1,1)$ and $(0,c_2)$ for some $0\le c_1\le 1,  0\le c_2\le 1$.\\
\textit{Class 2}: the trade boundary touches on the value profiles $(0,c_1)$ and $(1,c_2)$ for some $0\le c_1\le 1,  0\le c_2\le 1$.\\
\textit{Class 3}: the trade boundary touches on the value profiles $(c_1,0)$ and $(c_2,1)$ for some $0\le c_1\le 1,  0\le c_2\le 1$.\\
\textit{Class 4}: the trade boundary touches on the value profiles $(c_1,0)$ and $(1,c_2)$ for some $0< c_1\le 1,  0\le c_2< 1$\footnote{The cases where $c_1=0$ are included in \textit{Class 2}, and the cases where $c_2=1$ are included in \textit{Class 3}. }.\\
\indent By $(iv)$ of Proposition \ref{p1}, I can show that the profit from the four vertices $(0,0),(0,1),(1,0), (1,1)$ will never be strictly positive for any mechanism in  \textit{Class 1}, \textit{Class 2} or \textit{Class 3}. To see this, note that  for any mechanism in  \textit{Class 1}: $t(0,0)=0-c_2=-c_2\le 0$, $t(0,1)=0$, $t(1,0)=(1-0)\cdot 1-1-1=-1<0$, $t(1,1)=(1-1)\cdot 1-(1-c_1)\cdot 1=-(1-c_1)\le 0$; for any mechanism in \textit{Class 2}: $t(0,0)=0-c_1=-c_1\le 0$, $t(0,1)=0$, $t(1,0)=(1-0)\cdot 1-1-c_2=-c_2\le 0$, $t(1,1)= 0$;  for any mechanism in \textit{Class 3}: $t(0,0)=0$, $t(0,1)=0$, $t(1,0)=(1-0)\cdot 1-(1-c_1)-1=-(1-c_1)\le 0$, $t(1,1)= 0-(1-c_2)=-(1-c_2)\le 0$.\\
\indent When $M_B+M_S\le 1$, consider the joint distribution that puts probability  masses $M_B$, $M_S$ and $1-M_B-M_S$ on the value profiles (1,0), (0,1) and (0,0) respectively. It is straightforward  to verify that this is a feasible value distribution; moreover, the profit under this joint distribution cannot be strictly positive.  When $M_B+M_S> 1$,  consider the joint distribution that puts probability masses of  $1-M_S$, $1-M_B$ and $M_B+M_S-1$ on the value profiles (1,0), (0,1) and (1,1) respectively. It is straightforward to verify that this is a feasible value distribution; moreover,  the profit under this joint distribution cannot be strictly positive. Therefore, I can restrict attention to \textit{Class 4} only.\\
\textbf{\textit{Step 2: Identify an upper bound of the profit guarantee.}}\\
\indent I propose a relaxation of \eqref{D} by ignoring many constraints. Specifically, the only remaining ones are the constraints for  four vertices (0,0), (1,0), (0,1) and (1,1) as well as two value profiles $(c_1,0)$ and $(1,c_2)$ on the trade boundary. Formally, I have the following relaxed problem: \[
\max_{\lambda_B,\lambda_S,\mu \in \mathcal{R}}\lambda_BM_B +\lambda_S M_S+\mu \tag{D'}\label{D'}\]subject to
\[\label{eq76}
    \mu \le 0, \tag{C.1.1}
\]
\[\label{eq77}
    \lambda_B c_1 +\mu \le 0, \tag{C.1.2}
\]
\[\label{eq78}
    \lambda_B +\lambda_S c_2 +\mu \le 0, \tag{C.1.3}
\]
\[\label{eq79}
    \lambda_S +\mu \le 0, \tag{C.1.4}
\]
\[\label{eq80}
    \lambda_B +\lambda_S +\mu \le 0, \tag{C.1.5}
\]
\[\label{eq81}
    \lambda_B +\mu \le c_1-c_2. \tag{C.1.6}
\]
Note that the value of \eqref{D'}, denoted by $val$\eqref{D'},  is weakly greater than the value of \eqref{D}. Now I will find an upper bound of the value of \eqref{D'} across $0<c_1\le 1$ and $0\le c_2<1$,  and show that it is attainable by constructing deterministic DSIC and EPIR trading mechanisms and a feasible value distribution.   I discuss four cases:\\
\textit{Case 1}: $\lambda_B\le 0, \lambda_S\le 0$. Note that  by \eqref{eq76}, $val$\eqref{D'}$\le 0$ for any $0<c_1\le 1$ and $0\le c_2<1$.\\
\textit{Case 2}: $\lambda_B\ge 0, \lambda_S\ge 0$. Note that  by \eqref{eq76}, \eqref{eq80} and that $M_B>M_S$, 
\[
    \begin{split}
       \lambda_BM_B +\lambda_S M_S+\mu & \le (\lambda_B+\lambda_S)M_B+\mu \\
       & = (\lambda_B+\lambda_S+\mu)M_B+\mu(1-M_B)\\ 
       & \le 0.
    \end{split}
\]
Thus, $val$\eqref{D'}$\le 0$ for any $0<c_1\le 1$ and $0\le c_2<1$.\\
\textit{Case 3}: $\lambda_B\le 0, \lambda_S\ge 0$. By the same argument as in \textit{Case 2}, $val$\eqref{D'}$\le 0$ for any $0<c_1\le 1$ and $0\le c_2<1$.\\
\textit{Case 4}: $\lambda_B\ge 0, \lambda_S\le 0$. If $c_1=1$, then $
       \lambda_BM_B +\lambda_S M_S+\mu  \le \lambda_B+\mu  \le 0$
 where the first inequality follows from  $\lambda_B\ge 0 \ge \lambda_S$ and the second inequality follows from \eqref{eq77}. If $c_2=0$, then $
       \lambda_BM_B +\lambda_S M_S+\mu \le \lambda_B+\mu  
       \le 0$ 
 where the first inequality follows from  $\lambda_B\ge 0 \ge \lambda_S$ and the second inequality follows from \eqref{eq78}. If $c_1\le c_2$, then $
       \lambda_BM_B +\lambda_S M_S+\mu \le \lambda_B+\mu  
       \le 0$  where the first inequality follows from  $\lambda_B\ge 0 \ge \lambda_S$ and the second inequality follows from \eqref{eq81}.
Therefore,  I can restrict attention to $0<c_2 < c_1<1$, because otherwise  the profit guarantee cannot be strictly positive.  Now I am left with \eqref{eq77}, \eqref{eq78} and \eqref{eq81} as they imply the other three constraints. Note that at least one of \eqref{eq77}, \eqref{eq78} and \eqref{eq81} is binding, otherwise I can increase the value of \eqref{D'}, for example,  by increasing $\lambda_B$ by a small amount. I thus discuss three situations:\\
$(a): \lambda_B c_1 +\mu = 0$.\\
\indent Plugging $\lambda_B=-\frac{\mu}{c_1}$ into \eqref{eq78} and \eqref{eq81}, I obtain that
\[\label{eq82}
    \lambda_S \le \frac{1-c_1}{c_1c_2}\mu,\tag{C.1.7}
\]
\[\label{eq83}
    \mu \ge -\frac{c_1(c_1-c_2)}{1-c_1}. \tag{C.1.8}
\]
Then I have that
\[
    \begin{split}
       \lambda_BM_B +\lambda_S M_S+\mu & = -\frac{\mu}{c_1}M_B+\lambda_S M_S+\mu \\
       & \le -\frac{\mu}{r_1}M_B+\frac{1-c_1}{c_1c_2}\mu M_S+\mu\\ 
       & = (-\frac{M_B}{c_1}+1+\frac{1-c_1}{c_1c_2} M_S)\mu \\ & \le \max\{0, \frac{c_1-c_2}{1-c_1}M_B-\frac{c_1-c_2}{c_2}M_S-\frac{c_1(c_1-c_2)}{1-c_1}\}, 
    \end{split}
\]
where the first inequality follows from \eqref{eq82} and the second inequality follows from \eqref{eq83}.\\
$(b): \lambda_B +\mu =c_1-c_2$.\\
\indent Plugging $\lambda_B=c_1-c_2-\mu$ into \eqref{eq77} and \eqref{eq78},  I obtain that 
\[\label{eq84}
    \lambda_S \le -\frac{c_1-c_2}{c_2}\mu,\tag{C.1.9}
\]
\[\label{eq85}
    \mu \le -\frac{c_1(c_1-c_2)}{1-c_1}.\tag{C.1.10}
\]
Then I have that
\[
    \begin{split}
       \lambda_BM_B +\lambda_S M_S+\mu & = (c_1-c_2-\mu)M_B+\lambda_S M_S+\mu \\
       & = (c_1-c_2)M_B+\lambda_S M_S+(1-M_B)\mu 
       \\ & \le  \frac{c_1-c_2}{1-c_1}M_B-\frac{c_1-c_2}{c_2}M_S-\frac{c_1(c_1-c_2)}{1-c_1},
    \end{split}
\]
where the inequality follows from \eqref{eq84} and \eqref{eq85}.\\
$(c): \lambda_B + \lambda_S c_2 +\mu =0$.\\
\indent Plugging $\lambda_B=-\mu-\lambda_S c_2$ into \eqref{eq77} and \eqref{eq81}, I obtain that 
\[
    \lambda_S \ge -\frac{c_1-c_2}{c_2}, \tag{C.1.11}\label{c111}
\]
\[
    \mu \le \frac{c_1c_2}{1-c_1}\lambda_S. \tag{C.1.12}\label{c112}
\]
Then I have that
\[
    \begin{split}
       \lambda_BM_B +\lambda_S M_S+\mu & = (-\lambda_S c_2-\mu)M_B+\lambda_S M_S+\mu \\
       & = (M_S-c_2M_B)\lambda_S+(1-M_B)\mu \\
       & \le  [M_S-c_2M_B+\frac{c_1c_2}{1-c_1}(1-M_B)]\lambda_S
       \\ & \le   \max\{0, \frac{c_1-c_2}{1-c_1}M_B-\frac{c_1-c_2}{c_2}M_S-\frac{c_1(c_1-c_2)}{1-c_1}\},
    \end{split}
\]
where the first inequality follows from \eqref{c112} and the second inequality follows from \eqref{c111}.\\
\indent Let $K(c_1,c_2):=\frac{c_1-c_2}{1-c_1}M_B-\frac{c_1-c_2}{c_2}M_S-\frac{c_1(c_1-c_2)}{1-c_1}$. I am now solving for $\max_{0<c_2<c_1<1}K(c_1,c_2)$. Fix an arbitrary $c_2\in (0,1)$.  Taking the  first order derivative w.r.t. $c_1$, I obtain that 
\[\label{eq88}
    \frac{\partial K(c_1,c_2)}{\partial c_1}=
    \frac{1}{(1-c_1)^2}[(1-c_2)M_B+(1-\frac{M_S}{c_2})(1-c_1)^2-1-c_2]. \tag{C.1.13}
\]
Let $Q(c_1,c_2):=(1-c_2)M_B+(1-\frac{M_S}{c_2})(1-c_1)^2-1-c_2$. If $M_S\ge c_2$, then $Q(c_1,c_2)\le Q(1,c_2)\le 0$, so $K(c_1,c_2)\le K(c_2,c_2)=0$ for any  $c_1\in (c_2,1)$. If $M_S < c_2$, then $Q (c_1,c_2)$ is decreasing in $c_1$. Note that
\[
\begin{split}
    Q(c_2,c_2) & = (1-c_2)M_B-\frac{M_S}{c_2}(1-c_2)^2-c_2(1-c_2)+c_2 \\
       & = (1-c_2)(M_B+M_S-(\frac{M_S}{c_2}+c_2)).
\end{split}
\]
Note that $M_B+M_S-(\frac{M_S}{c_2}+c_2)\le M_B+M_S-2\sqrt{M_S}$. Therefore, if $M_B+M_S-2\sqrt{M_S}\le 0$, $Q(c_1,c_2)\le 0$ for any  $c_1\in(c_2,1)$. Then $K(c_1,c_2)\le K(c_2,c_2)=0$ for any  $c_1\in (c_2,1)$. If  $M_B+M_S-2\sqrt{M_S} > 0$, solving for $Q(c^*_1(c_2),c_2)=0$, I obtain that (I ignore the other solution which exceeds 1)
\[\label{eq89}
   c^*_1(c_2)=1-\sqrt{\frac{(1-c_2)(1-M_B)}{1-\frac{M_S}{c_2}}}.\tag{C.1.14} 
\]
If $c^*_1(c_2)\le c_2$, then again $K(c_1,c_2)\le K(c_2,c_2)=0$ for any  $c_1\in (c_2,1)$. Now if 
$ c^*_1(c_2)> c_2$, 
which, by some algebra, is equivalent to  that
\[\label{eq91}
    \sqrt{(1-c_2)(1-\frac{M_S}{c_2})}> \sqrt{1-M_B}, \tag{C.1.15}
\]
then $K(c_1,c_2)$ is maximized at $c_1=c^*_1(c_2)$, and the maximized value  $K(c_1^*(c_2),c_2)$, by some algebra, is equal to
\[\label{eq92}
    (\sqrt{(1-c_2)(1-\frac{M_S}{c_2})}-\sqrt{1-M_B})^2. \tag{C.1.16}
\]
Note that \eqref{eq92} is maximized at $c_2^*=\sqrt{M_S}$. And  $c_1^*(c_2^*)=1-\sqrt{1-M_B}$. Note that $c_1^*(c_2^*)>c_2^*$ is equivalent to that $\sqrt{1-M_B}+\sqrt{M_S}<1 $, which, by some algebra, is equivalent to that $M_B+M_S-2\sqrt{M_S} > 0$. Thus, I have found the solution $c_1^*=c_1^*(c_2^*)=1-\sqrt{1-M_B}$ and $c_2^*=\sqrt{M_S}$ when $M_B+M_S-2\sqrt{M_S} > 0$. And the maximized value  $K(c_1^*,c_2^*)=(1-\sqrt{M_S}-\sqrt{1-M_B})^2$.\\
\textbf{\textit{Step 3: Show that the upper bound is attainable.}}\\
\indent The last step is to construct deterministic trading mechanisms whose profit guarantee is $(1-\sqrt{M_S}-\sqrt{1-M_B})^2$ when $M_B+M_S-2\sqrt{M_S} > 0$ (equivalently, when $\sqrt{1-M_B}+\sqrt{M_S}<1$). Consider any deterministic trading mechanism satisfying $(i)$, $(ii)$ and $(iii)$ of Theorem \ref{t3}. Let
$\lambda_B=\frac{1-\sqrt{1-M_B}-\sqrt{M_S}}{\sqrt{1-M_B}}$,$
\lambda_S=-\frac{1-\sqrt{1-M_B}-\sqrt{M_S}}{\sqrt{M_S}}$ and $\mu=-\frac{(1-\sqrt{1-M_B}-\sqrt{M_S})(1-\sqrt{1-M_B})}{\sqrt{1-M_B}}$. I will show that they are feasible for the original dual problem \eqref{D}.\\
\indent Note that the constraint for the value profile (1,0) holds with equality by construction. Then the constraint holds for any \textit{interior} value profile\footnote{A value profile in which trade takes place with probability one is referred to as an interior value profile.}. Indeed,  the constraint is the most stringent for the value profile (1,0) because the trade boundary is non-decreasing. To see this, note that the constraint for any interior value profile $(v_B,v_S)$ is equivalent to that
\[\label{eq94}
    \lambda_Bv_B+b_1(v_B)+\lambda_S v_S - b_2(v_S)+\mu \le 0,\tag{C.1.17}
\]
where $(v_B,b_1(v_B))$ and $(b_2(v_S),v_S)$ are in the trade boundary. Then the L.H.S. of \eqref{eq94} is maximized at (1,0) because that $\lambda_B >0>\lambda_S$ and that $b_1$ as well as $b_2$ are non-decreasing. For any \textit{exterior} value profile\footnote{A value profile in which trade does not take place is referred to as an exterior value profile. Note that by the definition of the trade boundary, a value profile on the trade boundary is also an exterior value profile.}, the constraint also holds if $(ii)$ of Theorem \ref{t3} holds. To see this, note that  given the constructed $\lambda_B$, $\lambda_S$ and $\mu$,  $\lambda_B v_B + \lambda_S v_S+\mu=0$ for the value profiles $(1-\sqrt{1-M_B},0)$ and  $(1,\sqrt{M_S})$. Then, by linearity,  $\lambda_B v_B + \lambda_S v_S+\mu=0$ for any value profile  on the line linking $(1-\sqrt{1-M_B},0)$ and  $(1,\sqrt{M_S})$.  Therefore, if $(ii)$ of Theorem \ref{t3} holds, the constraint also holds for any exterior value profile because that $\lambda_B >0>\lambda_S$ and that the trade boundary is non-decreasing. Finally,  the value of \eqref{D} under the constructed dual variables   is $(1-\sqrt{M_S}-\sqrt{1-M_B})^2$ by simple calculation.\\
\indent Now consider the joint distribution described in Theorem \ref{t3}.   First,  it is straightforward to verify that it is a feasible value distribution. Second, given any trading mechanism satisfying $(i)$, $(ii)$ and $(iii)$ of Theorem \ref{t3}, the value of \eqref{p} is
$(1-\sqrt{M_S}-\sqrt{1-M_B})^2$ under the joint distribution by simple algebra. This finishes the proof.
\begin{remark}
\normalfont If $M_S=0$, then it is common knowledge that $v_S=0$. It is straightforward to show that the trading mechanism in which trade takes place with probability one if and only if $v_B> 1-\sqrt{1-M_B}$ is a maxmin deterministic trading mechanism. The worst-case value distribution puts a probability mass of $\sqrt{1-M_B}$ and $1-\sqrt{1-M_B}$ on the value profile $(1-\sqrt{1-M_B},0)$ and $(1,0)$ respectively.  If $M_B=1$, then it is common knowledge that $v_B=1$. It is straightforward to show that the trading mechanism in which trade takes place with probability one if and only if $v_S< \sqrt{M_S}$ is a maxmin deterministic trading mechanism. The worst-case value distribution puts  probability masses of $\sqrt{M_S}$ and $1-\sqrt{M_S}$ on the value profiles $(1,\sqrt{M_S})$ and $(1,0)$ respectively.
\end{remark}
\section{Proofs for Section \ref{s7}}\label{ad}
\subsection{Proof of Theorem \ref{t4}}
Theorem \ref{t4} has already been shown in Section \ref{s71}.
\subsection{Proof of Theorem \ref{t6}}
First, consider the following  auxiliary  problem:
 \[
\sup_{\pi\in \Delta(V)}\int_v[(v_B-v_S)\mathbb{1}_{v_B\ge v_S}-t^*(v)]d\pi(v) \tag{AP-1}\label{ap1}\]
subject to \[(q^*,t_B^*,t_S^*)\in \arg\sup_{(q,t_B,t_S)\in \mathcal{D}}\int_vt(v)d\pi(v).\]
\begin{lemma}\label{l3}
The value of the problem \eqref{ap1} is weakly higher than that of the problem \eqref{fi}.
\end{lemma}
\begin{proof}
This follows from  trading probability being bounded from below by 0 and bounded from above by 1.
\end{proof}
Now consider another auxiliary problem: \[
\sup_{\pi\in \Delta(V)}\int[(v_B-v_S)\mathbb{1}_{v_B\ge v_S}-t^*(v)]d\pi(v) \tag{AP-2}\label{ap2}\]
subject to \[(q^*,t_B^*,t_S^*)\in \arg\inf_{(q,t_B,t_S)\in \mathcal{D}}\int[(v_B-v_S)\mathbb{1}_{v_B\ge v_S}-t(v)]d\pi(v).\]
\begin{lemma}\label{l4}
The problems \eqref{ap1} and \eqref{ap2} are equivalent.
\end{lemma}
\begin{proof}
This follows from   their objective functions being the same and   their constraints being equivalent.
\end{proof}
\begin{lemma}\label{l5}
The symmetric triangular value distribution with $r=\frac{1}{2}$ is a solution to \eqref{ap2}. In addition, the value of \eqref{ap2} is $\frac{1}{2}$.
\end{lemma}
\begin{proof}
\textit{Step 1}: The random double auction with $r=\frac{1}{2}$ minimizes $\int[(v_B-v_S)\mathbb{1}_{v_B\ge v_S}-t(v)]d\pi(v)$ under the symmetric triangular value distribution with $r=\frac{1}{2}$.  To show this, note that it is equivalent to showing that  the random double auction with $r=\frac{1}{2}$ maximizes the intermediary's expected profit under the symmetric triangular value distribution with $r=\frac{1}{2}$. This is a direct implication of \textit{Step 1} in the proof of Theorem \ref{t1}.\\
\textit{Step 2}:  The symmetric triangular value distribution $r=\frac{1}{2}$ maximizes $\int[(v_B-v_S)\mathbb{1}_{v_B\ge v_S}-t(v)]d\pi(v)$ under the random double auction with $r=\frac{1}{2}$. To show this, note that  under the random double auction with $r=\frac{1}{2}$, I have that  \[
(v_B-v_S)\mathbb{1}_{v_B\ge v_S}-t(v)=
\left\{
\begin{array}{rcl}
\frac{1}{2}     &      & {\normalfont \text{if}\quad v_B-v_S\ge \frac{1}{2},}\\
v_B-v_S      &      & {\normalfont \text{if}\quad 0<v_B-v_S<\frac{1}{2},}
\\
0     &      & {\normalfont \text{if}\quad v_B-v_S\le 0.}
\end{array} \right.\]
Therefore, any value distribution whose support is contained in  $\{v\in V|v_B-v_S\ge \frac{1}{2}\}$ maximizes $\int[(v_B-v_S)\mathbb{1}_{v_B\ge v_S}-t(v)]d\pi(v)$ under the random double auction with $r=\frac{1}{2}$. It is straightforward that the symmetric triangular value distribution with $r=\frac{1}{2}$ is such a value distribution.\\
\indent The first statement of this lemma follows immediately from these two steps and the properties of  a saddle point.\\
\indent Finally, the value of the problem \eqref{ap2} is $\frac{1}{2}$ by simple calculation.
\end{proof}
\begin{lemma}\label{l6}
The value of the problem \eqref{fi} is at least $\frac{1}{2}$. 
\end{lemma}
\begin{proof}
Recall that given the the symmetric triangular value distribution with $r=\frac{1}{2}$,  the intermediary is indifferent between trading and no trading for any value profile in the support except for $(1,0)$, so an \textit{efficient} trading mechanism\footnote{Precisely, trade takes place with probability one for any report in the support of the symmetric triangular value distribution with $r=\frac{1}{2}$, and 0 otherwise.} is an optimal trading mechanism. The value of the objective function under the symmetric triangular value distribution with $r=\frac{1}{2}$ and the efficient trading mechanism is $\frac{1}{2}$ by simple calculation.
\end{proof}
\indent Theorem \ref{t6} follows immediately from Lemma \ref{l3}, \ref{l4}, \ref{l5} and \ref{l6}.
\section{Illustration of Theorem \ref{t2}}\label{ae}
\subsection{Characterization of Logarithmic Trading Mechanism}
For the asymmetric case, it is natural to attach different weights to the submitted bid price and the submitted ask price.  I thus  form an educated guess of the trading region in a maxmin trading mechanism: trade takes place with positive probability if and only if the difference between a \textit{weighted}  bid (true value of the buyer) $r_2\cdot v_B$ and a (different) \textit{weighted}  ask (true value of the seller) $(1-r_1)\cdot v_S$ exceeds a threshold $r_1r_2>0$, or $r_2v_B-(1-r_1)v_S>r_1r_2$. In addition, again,  the support of a worst-case value distribution coincides with the trading region (including the boundary). Together with $(iv)$ of Proposition \ref{p1}, the complementary slackness condition (\ref{cs}) can be expressed as follows: for any $(v_B,v_S)\in AT$,
\[ 
\lambda_B^{**}v_B+\lambda_S^{**}v_S+\mu^{**} = (v_B-v_S)q^{**}(v_B,v_S)-\int_{\frac{1-r_1}{r_2}v_S+r_1}^{v_B} q^{**}(x,v_S)dx-\int_{v_S}^{\frac{r_2}{1-r_1}(v_B-r_1)} q^{**}(v_B,x)dx.\tag{CS-2}\label{eq29}
\]
Now I solve for  the trading rule $q^{**}$. Similarly,  I  first take the first order derivatives with respect to $v_B$ and $v_S$ respectively, and I obtain that for any $(v_B,v_S)\in AT$,
\[ \label{eq30}
(v_B-v_S)\frac{\partial q^{**}(v_B,v_S)}{\partial v_B}-\frac{\partial \int_{v_S}^{\frac{r_2}{1-r_1}(v_B-r_1)} q^{**}(v_B,x)dx }{\partial v_B}=\lambda_B^{**}, \tag{FOC-B-2}
\]
\[ \label{eq31}
(v_B-v_S)\frac{\partial q^{**}(v_B,v_S)}{\partial v_S}-\frac{\partial \int_{\frac{1-r_1}{r_2}v_S+r_1}^{v_B} q^{**}(x,v_S)dx }{\partial v_S}=\lambda_S^{**}. \tag{FOC-S-2}
\]
Then, I take the cross partial derivative, with some algebra, I obtain that
\[ \label{eq32}
(v_B-v_S)\frac{\partial q^{**}(v_B,v_S)}{\partial v_B\partial v_S}=0.
\]
Thus, $q^{**}(v_B,v_S)$ is separable,  which can be expressed  as the sum of two functions $f^{**}$ and $g^{**}$: for any $(v_B,v_S)\in AT$,
\[\label{eq33}
    q^{**}(v_B,v_S)=f^{**}(v_B)+g^{**}(v_S).\tag{E.1.1}
\]
Again, the separable nature is crucial for solving \eqref{eq29}.
Plugging \eqref{eq33} into \eqref{eq30} and \eqref{eq31}, I obtain that  for any $(v_B,v_S)\in AT$,
\[\label{eq34}
[(1-\frac{r_2}{1-r_1})v_B+\frac{r_1r_2}{1-r_1}](f^{**})'(v_B)-\frac{r_2}{1-r_1}[f^{**}(v_B)+g^{**}(\frac{r_2}{1-r_1}(v_B-r_1))]=\lambda_B^{**},\tag{E.1.2}
\]
\[\label{eq35}
[(\frac{1-r_1}{r_2}-1)v_S+r_1](g^{**})'(v_S)+\frac{1-r_1}{r_2}[f^{**}(\frac{1-r_1}{r_2}v_S+r_1)+g^{**}(v_S)]=\lambda_S^{**}.\tag{E.1.3}
\]
Note that $f^{**}(v_B)+g^{**}(\frac{r_2}{1-r_1}(v_B-r_1))=0$ and  that $f^{**}(\frac{1-r_1}{r_2}v_S+r_1)+g^{**}(v_S)=0$ because trade does not take place in the boundary of the trading region, i.e.,  $q^{**}(v_B,v_S)=0$ for $r_2v_B-(1-r_1)v_S=r_1r_2$.  Then it is straightforward to  solve for $f^{**}(v_B)$ and $g^{**}(v_S)$, and I obtain that 
\[\label{eq36}
    f^{**}(v_B)=\frac{(1-r_1)\lambda_B^{**}}{1-r_1-r_2}\ln{[(1-\frac{r_2}{1-r_1})v_B+\frac{r_1r_2}{1-r_1})]}+c_B^{**},\tag{E.1.4}
\]
\[\label{eq37}
    g^{**}(v_S)=\frac{r_2\lambda_S^{**}}{1-r_1-r_2}\ln{[(\frac{1-r_1}{r_2}-1)v_S+r_1]}+c_S^{**},\tag{E.1.5}
\]
where $c_B^{**}$ and $c_S^{**}$ are some constants. Observe that $$
g^{**}(\frac{r_2(v_B-r_1)}{1-r_1})=\frac{r_2\lambda^{**}_S}{1-r_1-r_2}\ln{[(1-\frac{r_2}{1-r_1})v_B+\frac{r_1r_2}{1-r_1}]}+c_S^{**}.$$
Then, again, using  that $q^{**}(v_B,v_S)=0$ for $r_2v_B-(1-r_1)v_S=r_1r_2$,  I  have  that
\[\label{eq38}
   (1-r_1)\lambda_B^{**}+r_2\lambda_S^{**}=c_B^{**}+c_S^{**}=0.\tag{E.1.6}
\]
Now plugging \eqref{eq36},\eqref{eq37} and \eqref{eq38} into \eqref{eq33}, I obtain that for any $(v_B,v_S)\in AT$,
\[\label{eq39}
    q^{**}(v_B,v_S)=\frac{(1-r_1)\lambda_B^{**}}{1-r_1-r_2}[\ln(\frac{1-r_1-r_2}{1-r_1}v_B+\frac{r_1r_2}{1-r_1})-\ln(\frac{1-r_1-r_2}{r_2}v_S+r_1)].
\]
Likewise,  to solve for $\lambda_B^{**}$, I let $q^{**}(1,0)$ be $1$ and obtain that   $\lambda_B^{**}=\frac{1-r_1-r_2}{(1-r_1)\ln{\frac{1-r_2}{r_1}}}$. So far I have obtained the trading rule $q^{**}$\footnote{Plugging the trading rule $q^{**}$ into \eqref{eq29}, it is straightforward that $\mu^{**}=-\frac{r_1(1-r_1-r_2)}{(1-r_1)\ln{\frac{1-r_2}{r_1}}}$.}. 
The payment rule $t_B^{**}$ (resp, the transfer rule $t_S^{**}$) is then characterized by $(ii)$ (resp, $(iii)$) of Proposition \ref{p1}. 
\subsection{Characterization of Asymmetric Triangular Value Distribution}
Similar to the symmetric case, I impose a \textit{0-weighted-virtual-value condition} on the joint distribution, stating that weighted virtual value is 0 for any value profile in the support except for the highest joint type. Formally, \[
  \Phi(v)=0,  \quad \forall v\in AT\backslash\{(1,0)\}.\tag{ZWVV-2} \label{zw2}
\] The construction procedure for the joint distribution is exactly the same. Therefore I omit it. Note that the marginal distribution no longer has a uniform distribution part since $v_B-v_S$ is no longer a constant on the  boundary of the trading region due to different weights for the buyer and the seller. The final step is  to make sure that the constructed joint distribution has the known expectations.  Given the marginal distributions for the buyer and the seller, I have a system of two equations \eqref{eq27} and \eqref{eq28}. Lemma \ref{l2} states that a solution exists  for the asymmetric case, details of which are given in  Appendix \ref{ab}.
\section{The Expectation of Buyer's Value is 1}\label{af}
Note that in this case, it is common knowledge that  $v_B=1$. The mechanism now consists of three one-dimensional functions: the trading rule, the payment rule and the transfer rule depend on the seller's submitted ask price only.  Consider the following mechanism. If $a<1-\iota$,
\[
q(a)=
    1-\frac{\ln{(1-a)}}{\ln{\iota}},\]
    \[t_B(a)=1-\frac{\ln{(1-a)}}{\ln{\iota}},\]
    \[t_S(a)=1+\frac{1-\iota-\ln{(1-a)}-a}{\ln{\iota}},\]
where $\iota\in (0,1)$ is the unique solution to \[\iota \ln{\iota}+1-\iota=M_S.\]
If $a\ge 1-\iota$,
\[q(a)=t_B(a)=t_S(a)=0.\] 
Consider the following  distribution of the seller's value: if $v_S\in (0,1-\iota]$, the density function $\pi_S(v_S)=\frac{\iota}{(1-v_S)^2}$; if $v_S=0$, the probability mass $Pr_S(0)=\iota$. \\
\indent It is straightforward to show that this mechanism is a maxmin trading mechanism and this distribution is a worst-case value distribution by a similar argument to \cite{carrasco2018optimal}. 
\begin{remark}
\normalfont Recall the logarithmic trading mechanism. Note that  $q^{**}(1,a)=\frac{1}{\ln\frac{1-r_2}{r_1}}[\ln{(1-r_2)}-\ln(\frac{1-r_1-r_2}{r_2}a+r_1)]$. When $r_1=1$, it is straightforward that it is reduced to the above mechanism. Also, When $v_B=r_1=1$, it is straightforward that  the asymmetric triangular value distribution is reduced to the above distribution. 
\end{remark}

\bibliographystyle{apalike}
\bibliography{abc}


%
%
%

\end{document}